%% file: main.tex
\documentclass[runningheads]{llncs}
\usepackage[hidelinks]{hyperref}

\usepackage{color}         

\usepackage{appendix}
\usepackage{thmtools,thm-restate}

\usepackage[utf8]{inputenc} 
\usepackage[T1]{fontenc}    

\usepackage{url}            
\usepackage{booktabs}       
\usepackage{amsfonts}       
\usepackage{nicefrac}       
\usepackage{microtype}      

\usepackage{algorithmic}
\usepackage{graphicx}
\usepackage{textcomp}
\usepackage{mathtools}

\usepackage{amssymb}
\usepackage{siunitx}
\usepackage{bm}
\usepackage{multirow}
\usepackage[caption=false]{subfig}

\usepackage{amsmath}
\usepackage[english]{babel}
\usepackage{tabularx}
\usepackage{stmaryrd}
\usepackage{xspace}

\usepackage[ruled,vlined,linesnumbered]{algorithm2e}
\usepackage{cite}
\usepackage{wrapfig}
\usepackage[normalem]{ulem}
\usepackage[subtle,mathspacing=normal,indent=tight,leading=tight]{savetrees}
\usepackage[strings]{underscore}
\usepackage{etoolbox}
\usepackage{enumitem}

\newcommand{\zerodisplayskips}{%
  \setlength{\abovedisplayskip}{3pt}%
  \setlength{\belowdisplayskip}{3pt}%
  \setlength{\abovedisplayshortskip}{3pt}%
  \setlength{\belowdisplayshortskip}{3pt}}
\appto{\normalsize}{\zerodisplayskips}
\appto{\small}{\zerodisplayskips}
\appto{\footnotesize}{\zerodisplayskips}

\makeatletter
\newcommand{\nsuccprec}{\not\mathrel{\mathpalette\succ@prec{\succ\prec}}_{re}}
\newcommand{\succ@prec}[2]{\succ@@prec#1#2}
\newcommand{\succ@@prec}[3]{%
  \vcenter{\m@th\offinterlineskip
    \sbox\z@{$#1#3$}%
    \hbox{$#1#2$}\kern-0.4\ht\z@\box\z@
  }%
}
\makeatother

\newcommand{\revision}[1]{{\color{black}#1}}

\newcommand{\srs}{R} 
\newcommand{\resv}{r} 
\newcommand{\Trec}{\alpha}
\newcommand{\Tdur}{\beta}

\begin{document}
\title{
An STL-based Formulation of Resilience in Cyber-Physical Systems
\thanks{This is the full version of the paper under the same title accepted to FORMATS 2022.}}


\author{
Hongkai Chen\inst{1}
\and 
Shan Lin\inst{1}
\and 
Scott A. Smolka\inst{1} \and 
Nicola Paoletti\inst{2}
}
\authorrunning{H. Chen et al.}
\institute{
Stony Brook University, Stony Brook, USA\\
\email{\{hongkai.chen,shan.x.lin\}@stonybrook.edu}\\
\email{sas@cs.stonybrook.edu}
\and
Royal Holloway, University of London, UK\\
\email{nicola.paoletti@rhul.ac.uk}
}
\maketitle

\vspace{-2ex}
\input{abstract_NP.tex}

\input{introduction.tex}

\input{STL}
\input{resilience_v2.tex}

\input{casestudy_v3.tex}

\input{relatedwork_NP}
\input{conclusion.tex}

\input{acknowledgment}

\bibliography{mybibfile}

\input{appendix.tex}

\end{document}

%% file: abstract_NP.tex
\begin{abstract}

Resiliency is the ability to quickly recover from a violation and avoid future violations for as long as possible. Such a property is of fundamental importance for Cyber-Physical Systems (CPS), and yet, to date, there is no widely agreed-upon formal treatment of CPS resiliency. 
We present an STL-based framework for reasoning about resiliency in CPS in which resiliency has a syntactic characterization in the form of an \emph{STL-based Resiliency Specification} (SRS).  Given an arbitrary STL formula $\varphi$, time bounds $\Trec$ and $\Tdur$, the SRS of $\varphi$, $R_{\Trec,\Tdur}(\varphi)$, is the STL formula $\neg \varphi\mathbf{U}_{[0,\Trec]}\mathbf{G}_{[0,\Tdur)}\varphi$, specifying that recovery from a violation of $\varphi$ occur within time~$\Trec$~(\emph{recoverability}), and subsequently that $\varphi$ be maintained for duration~$\Tdur$~(\emph{durability}).
These $R$-expressions, which are atoms in our SRS logic, can be combined using STL operators, allowing one to express composite resiliency specifications, e.g.,  multiple SRSs must hold simultaneously, or the system must eventually be resilient.
We define a quantitative semantics for SRSs in the form of a \emph{Resilience Satisfaction Value} (ReSV) function $r$ and prove its soundness and completeness w.r.t.\ STL's Boolean semantics. 
The $r$-value for $R_{\Trec,\Tdur}(\varphi)$ atoms is a singleton set containing a pair quantifying recoverability and durability.
The $r$-value for a composite SRS formula results in a set of non-dominated recoverability-durability
pairs, given that the ReSVs of subformulas might not be directly comparable (e.g., one subformula has superior durability but worse recoverability than another). To the best of our knowledge, this is the first \emph{multi-dimensional} quantitative semantics for an STL-based logic.  Two case studies demonstrate the practical utility of our approach.

\end{abstract}

%% file: introduction.tex
\section{Introduction}

Resiliency (\emph{syn.}\  resilience) is defined as the ability to recover from or adjust easily to adversity or change~\cite{MW}.  Resiliency is of fundamental importance in Cyber-Physical Systems~(CPS), which are expected to exhibit safety- or mission-critical behavior even in the presence of internal faults or external disturbances. Consider for example the \emph{lane keeping} problem for autonomous vehicles~(AVs), which requires a vehicle to stay within the marked boundaries of the lane it is driving in at all times. The standard temporal-logic-based notion of safety is not ideally suited for specifying the AV's behavior when it comes to lane keeping.  This is because AV technology is not perfect and driving conditions~(e.g., being crowded by a neighboring vehicle) and other external disturbances may require occasional or even intermittent violations of lane keeping. 
Rather, the AV should behave resiliently in the presence of a lane violation, recovering from the violation in a timely fashion, and avoiding future lane departures for as long as possible. 
Unfortunately, there is no widely agreed notion of resiliency within the CPS community, despite several efforts to settle the issue~(see Section~\ref{sec:relatedwork}). 

\vspace{-2ex}
\paragraph{Our Contributions.} In this paper, we present
an STL-based framework for reasoning about resiliency in Cyber-Physical Systems. In our approach, resiliency has a syntactic characterization in the form of an \emph{STL-based Resiliency Specification}~(SRS). Given an arbitrary STL formula $\varphi$, time bounds~$\Trec$ and $\Tdur$, the SRS of $\varphi$, $R_{\Trec,\Tdur}(\varphi)$, is the STL formula $\neg \varphi\mathbf{U}_{[0,\Trec]}\mathbf{G}_{[0,\Tdur)}\varphi$, which 
specifies that recovery from a violation of $\varphi$ occur within time~$\Trec$, and subsequently $\varphi$ be maintained for duration~$\Tdur$ at least.
The SRS of $\varphi$ captures the requirement that a system quickly recovers from a violation of $\varphi$ (\emph{recoverability})  and then satisfy $\varphi$ for an extended period of time (\emph{durability}).
The $R_{\Trec,\Tdur}(\varphi)$ expressions, which are atoms in our SRS logic, can be inductively combined using STL operators, allowing one to express composite resiliency specifications; e.g., multiple SRSs must hold simultaneously~($R_{\Trec_1,\Tdur_1}(\varphi_1) \wedge R_{\Trec_2,\Tdur_2}(\varphi_2)$), or that the system must eventually be resilient~($\mathbf{F}_I R_{\Trec,\Tdur}(\varphi)$). 

We define a quantitative semantics for SRSs in the form of a \emph{Resilience Satisfaction Value} (ReSV) function $r$.
Our semantics for $R_{\Trec,\Tdur}(\varphi)$ atoms is a singleton set of the form $\{(\mathit{rec},\mathit{dur})\}$, where $\mathit{rec}$ quantifies how early before bound~$\Trec$ recovery occurs, and $\mathit{dur}$ indicates 
for how long after bound~$\Tdur$ property $\varphi$ is maintained. To the best of our knowledge, this is the first \emph{multi-dimensional} quantitative  semantics for STL.


Our approach does not make any simplifying assumption as to which of the two requirements (recoverability and durability) to prioritize or how to combine the two values. This decision can lead to a semantic structure involving two or more non-dominated $(\mathit{rec}, \mathit{dur})$ pairs.
In such situations, we choose to retain \emph{all} non-dominated pairs so as to provide a comprehensive, assumption-free, characterization of CPS resiliency. Thus, our semantics is a set of non-dominated $(\mathit{rec},\mathit{dur})$ pairs, which is derived inductively from subformulas using Pareto optimization.

For example, consider the SRS $\psi_1 \vee \psi_2$, where the ReSV of $\psi_1$~(over a given signal at a particular time) is $\{(2, 5)\}$ and, similarly, the ReSV of $\psi_2$ is $\{(3, 3)\}$. The semantics of $\psi_1 \vee \psi_2$ should choose the dominant pair, but the two are non-dominated: $(3,3)$ has better recoverability, while $(2,5)$ has better durability.  So we include both.
We prove that our semantics is sound and complete with respect to
the classic STL Boolean semantics by~(essentially) showing that an SRS $\psi$ has at least one non-dominated pair with $\mathit{rec}$, $\mathit{dur}$ $>0$ iff $\psi$ is true.

We perform an extensive experimental evaluation of our framework centered around two case studies: UAV package delivery and multi-agent flocking.  In both cases, we formulate mission requirements in STL, and evaluate their ReSV values in the context of various SRS specifications.  Our results clearly demonstrate the expressive power of our framework.

%% file: STL.tex
\section{Preliminaries}
\label{sec:prelim}

In this section, we introduce the syntax and semantics of Signal Temporal Logic (STL)~\cite{maler2004monitoring,donze2010robust}.
STL is a formal specification language for real-valued signals. We consider $n$-dimensional discrete-time signals \revision{$\xi:\mathbb{T}\rightarrow\mathbb{R}^n$}
where $\mathbb{T}= \mathbb{Z}_{\geq 0}$ is the (discrete) time domain.\footnote{Discrete-time signals over an arbitrary time step can always be mapped to signals over a unit time step.}
$\mathbb{T}$ is the interval $[0,|\xi|]$, where $|\xi|>0$ is the length of the signal. If $|\xi|<\infty$, we call $\xi$ bounded. We use the words signal and trajectory interchangeably. An STL atomic predicate $p\in \mathit{AP}$ is defined over signals and is of the form $p\equiv \mu(\xi(t)) \geq c$, $t\in\mathbb{T}$, $c\in\mathbb{R}$, and $\mu:\mathbb{R}^n\rightarrow\mathbb{R}$.  STL formulas $\varphi$ are defined recursively according to the following grammar~\cite{donze2010robust}:
\begin{align*}
    \varphi ::= 
    p\;|\; \neg\varphi\;|\; 
    \varphi_1 \wedge\varphi_2 \;|\; 
    \varphi_1 \mathbf{U}_I \varphi_2
\end{align*}

\noindent where $\mathbf{U}$ is the \emph{until} operator and $I$ is an interval on $\mathbb{T}$. \revision{Logical disjunction is derived from $\wedge$ and $\neg$ as usual, and }
operators \emph{eventually} and \emph{always} are derived from $\mathbf{U}$ as usual: $\mathbf{F}_{I}\varphi = \top\mathbf{U}_I \varphi$ and $\mathbf{G}_{I}\varphi = \neg (\mathbf{F}_I \neg\varphi)$. The satisfaction relation $(\xi,t)\models \varphi$, indicating $\xi$ satisfies $\varphi$ at time $t$, is defined as follows:\footnote{Given $t \in \mathbb{T}$ and interval $I$ on $\mathbb{T}$, $t + I$ is used to denote the set $\{t+t' \mid t' \in I\}$.}
\vspace{0.25ex}
\begin{align*}
&(\xi, t) \models p &\Leftrightarrow &\hspace{2ex} \mu(\xi(t)) \geq c\vspace{0.02in}\\
&(\xi, t) \models \neg \varphi &\Leftrightarrow &\hspace{2ex} \neg((\xi, t) \models \revision{\varphi})\vspace{0.02in}\\
&(\xi,t) \models \varphi_1 \wedge \varphi_2  &\Leftrightarrow&\hspace{2ex} (\xi,t) \models \varphi_1 \wedge (\xi,t) \models \varphi_2 \vspace{0.02in}\\
&(\xi,t) \models \varphi_1 \mathbf{U}_I \varphi_2 &\Leftrightarrow&\hspace{2ex} \exists\ t'\in t + I\ \text{s.t.}\ (\xi,t') \models \varphi_2 \wedge \forall\; t''\in [t,t'),\ (\xi,t'') \models \varphi_1 
\end{align*}
\vspace{0.25ex}
We call an STL formula $\varphi$ \emph{bounded-time} if all of its temporal operators are bounded (i.e., their intervals have finite upper bounds) and $|\xi|$ is large enough to determine satisfiability \revision{at time 0}; i.e., $|\xi|$ is greater than the maximum over the sums of all
the nested upper bounds on the temporal operators~\cite{yaghoubi2020worst}.  For example, if $\varphi$ is $\revision{\varphi_1}\mathbf{U}_{[0,5]}\mathbf{G}_{[1,2]}\varphi_2\wedge\mathbf{F}_{[0,10]}\mathbf{G}_{[1,6]}\varphi_2$, then a trajectory with length $N \geq \max(5+2, 10+6) = 16$ is sufficient to determine whether $\varphi$ holds.  In this paper, we only consider bounded-time STL formulas \revision{as in its original definition~\cite{maler2004monitoring}.}

STL admits a quantitative semantics given by a real-valued function $\rho$ such that $\rho(\varphi,\xi,t)>0 \Rightarrow (\xi,t) \models \varphi$, and defined as follows~\cite{donze2010robust}:
\vspace{0.25ex}
\begin{align*}
\rho(\revision{\mu(\xi(t)) \geq c}, \xi, t)&= \mu(\xi(t))-c   \vspace{0.02in}\\
\rho(\neg \varphi, \xi, t)&= -\rho(\varphi, \xi, t) \vspace{0.02in}\\
\rho(\varphi_1 \wedge \varphi_2, \xi, t) &= \min(\rho(\varphi_1, \xi, t),\rho(\varphi_2, \xi, t)) \vspace{0.02in}\\
\rho(\varphi_1\mathbf{U}_I\varphi_2, \xi, t) &= \max_{t'\in t+I}~\min (\rho(\varphi_2, \xi, t'),\min_{t''\in [t,t+t')}\rho(\varphi_1, \xi, t''))
\end{align*}

\vspace{-0.5ex}
A $\rho$-value, called the \emph{robustness satisfaction value}~(RSV), can be interpreted as the extent to which $\xi$ satisfies $\varphi$ at time $t$.  Its absolute value can be viewed as the distance of $\xi$ from the set of trajectories satisfying or
violating $\varphi$, with positive values indicating satisfaction and negative values indicating violation. 

%% file: resilience_v2.tex
\section{Specifying Resilience in STL }\label{sec:resilience}

In this section, we introduce our STL-based resiliency specification formalism and its quantitative semantics in terms of non-dominated recoverability-durability pairs.

\subsection{Resiliency Specification Language}
We introduce an STL-based temporal logic to reason about resiliency of STL formulas. 
Given an STL specification $\varphi$, there are two properties that characterize its resilience w.r.t.\ a signal $\xi$, namely, \emph{recoverability} and \emph{durability}: the ability to (1)~recover from a violation of $\varphi$ within time $\Trec$, and (2)~subsequently maintain $\varphi$ for at least time $\Tdur$.

\begin{example}\label{example:1}
Consider an STL specification $\varphi = (2\leq y \leq 4)$, where $y$ is a signal. In Figure~\ref{fig:stl_property}(a), signals $\xi_1$ and $\xi_2$ violate $\varphi$ at time $t_1$. Given recovery deadline $\alpha$, we see that only $\xi_1$ satisfies recoverability of $\varphi$ w.r.t.\ $\alpha$ because $\varphi$ becomes true before $t_1+\Trec$. In the case of $\xi_2$, $\varphi$ becomes true only after $t_1+\Trec$. In Figure~\ref{fig:stl_property}(b), signals $\xi_3$ and $\xi_4$ recover to satisfy $\varphi$ at time $t_2$.  Given durability bound $\beta$, we observe that only $\xi_3$ is durable w.r.t.\ $\beta$. 
\end{example}

\vspace{-5ex}
\begin{figure}[ht]
	\centering
	\subfloat[Signal $\xi_1$ satisfies recoverability.]{\includegraphics[width=.45\linewidth]{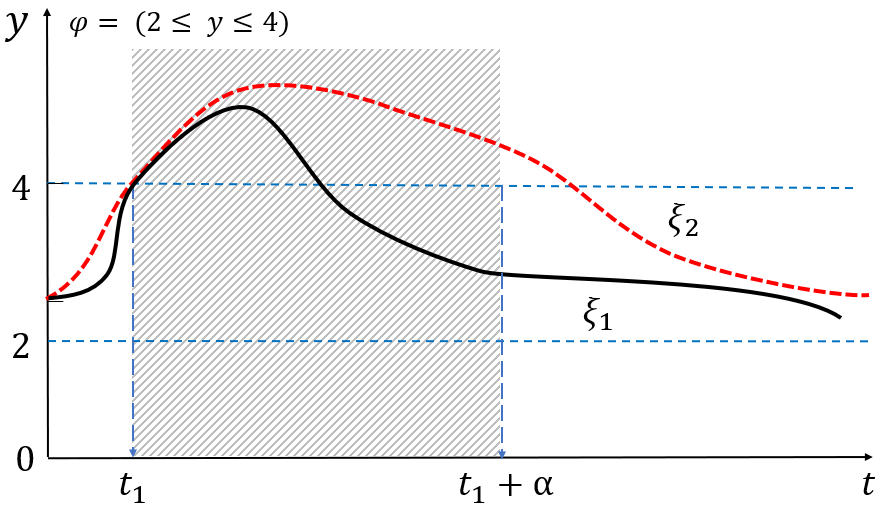}\label{fig:recoverability}}\hfill
	\subfloat[Signal $\xi_3$ satisfies durability.]{\includegraphics[width=.45\linewidth]{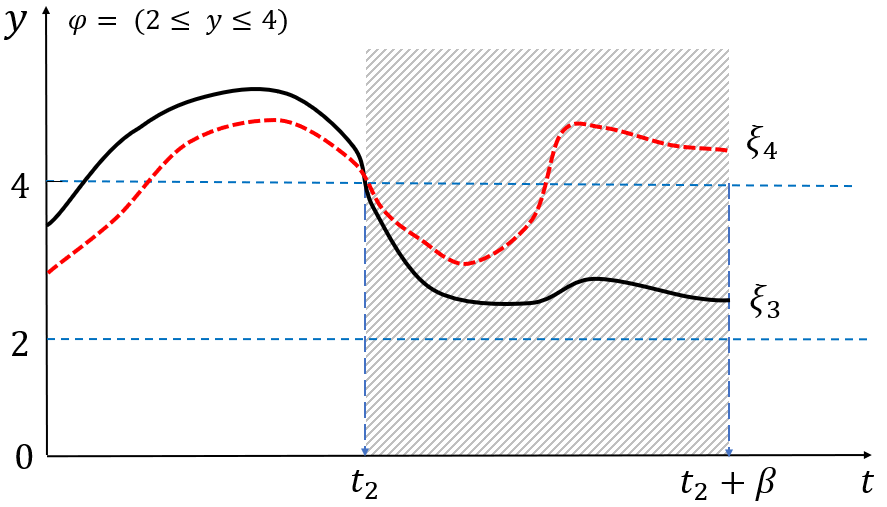}
		\label{fig:durability}}
	\vspace{-5pt}
	\caption{Resilience w.r.t.\ an STL formula $\varphi = (2\leq y \leq 4)$}
	\label{fig:stl_property}
\end{figure}

\vspace{-2ex}
Resilience of an STL formula $\varphi$ relative to a signal should be determined by the joint satisfaction of recoverability and durability of $\varphi$. We can formalize this notion using the STL formula $R_{\Trec,\Tdur}(\varphi) \equiv \neg \varphi\mathbf{U}_{[0,\Trec]}\mathbf{G}_{[0,\Tdur)}\varphi$, which captures the requirement that the system recovers from a violation of $\varphi$ within bound $\Trec$ and subsequently maintains $\varphi$ for bound $\Tdur$. In what follows, we introduce SRS, an STL-based resiliency specification language that allows one to combine multiple $R_{\Trec,\Tdur}(\varphi)$ expressions using Boolean and temporal operators.

\begin{definition}[STL-based Resiliency Specification]
The STL-based resiliency specification~(SRS) language is defined by the following grammar:
\vspace{0.25ex}
\begin{align*}
    \psi ::= \srs_{\Trec,\Tdur}(\varphi)
    \;|\; \neg\psi\;|\; 
    \psi_1 \wedge\psi_2 \;|\; \psi_1 \mathbf{U}_I \psi_2
\end{align*}

\noindent where $\varphi$ is an STL formula, $\srs_{\Trec,\Tdur}(\varphi) \equiv \neg \varphi\mathbf{U}_{[0,\Trec]}\mathbf{G}_{[0,\Tdur)}\varphi$, $\Trec,\Tdur\in \mathbb{T}$, 
$\Tdur>0$. 
\end{definition}

\revision{We use a semi-closed interval in the $\mathbf{G}$ operator, reflecting our requirement that $\varphi$ stays true for a time interval of duration $\beta$, and may be false at the end of the interval. The piece-wise constant interpretation of discrete-time signals implies that properties between two consecutive steps remain unchanged~(see Remark~\ref{rem:time_rob}).
} 
\revision{SRS formulas are also restricted to bounded-time intervals as those in STL formulas.}
Boolean satisfiability of an SRS formula $\psi$ reduces to the satisfiability of the corresponding STL formula obtained by replacing every atom $\srs_{\Trec,\Tdur}(\varphi)$ with $\neg \varphi\mathbf{U}_{[0,\Trec]}\mathbf{G}_{[0,\Tdur)}\varphi$. 
Note that satisfiability of an $\srs_{\Trec,\Tdur}(\varphi)$ atom involves satisfying the recoverability and durability requirements for $\varphi$ highlighted in Example~\ref{example:1}. 

\begin{remark}[Why a new logic? (1/2)]
We note that SRS is equivalent to STL in the sense that any STL formula $\varphi$ is equivalent to the SRS atom $\srs_{0,1}(\varphi)$ and, conversely, any SRS formula is an STL formula. \textit{Then why would we introduce a new logic?} The SRS logic is explicitly designed to express specifications that combine resiliency requirements with Boolean and temporal operators. Most importantly, as we will see, our semantics for SRS is defined inductively starting from $\srs_{\Trec,\Tdur}(\varphi)$ atoms, and not from STL atomic predicates. The former are more expressive than the latter~(an STL atomic predicate can be expressed via an SRS atom, but not vice versa).
\end{remark}

\begin{remark}[Resilience vs.\ safety, liveness, and stability]
We stress that safety properties $\mathbf{G}_I \varphi$ are not sufficient to capture resilience, as they do not allow for occasional, short-lived violations. Among the liveness properties, reachability properties $\mathbf{F}_I \varphi$ do not capture the requirement that we want to satisfy $\varphi$ in a durable manner. Similarly, progress properties, summarized by the template $\mathbf{G}_{I_1} (\neg \varphi \rightarrow \mathbf{F}_{I_2} \varphi)$, do not require the signal to satisfy $\varphi$ for extended time periods (but they  ``infinitely'' often lead to $\varphi$).

Thus, one might be tempted to combine safety and liveness and express resilience as a $\mathbf{F}_{[0,\Trec]}\mathbf{G}_{[0,\Tdur)} \varphi$, a template often called \emph{stability} or stabilization~\cite{cook2011proving}. There is a subtle but important difference between stability and our definition $\neg \varphi\mathbf{U}_{[0,\Trec]}\mathbf{G}_{[0,\Tdur)}\varphi$ of resilience: there could be multiple recovery episodes occurring in a trajectory, i.e., time steps where $\varphi$ transitions from false to true. Stability is satisfied by \textit{any} recovery episode within time $[0,\Trec]$ provided that $\varphi$ stays true for at least time $\Tdur$. On the contrary, resilience is satisfied by only the \textit{first} recovery episode, provided that $\varphi$ is not violated for longer than $\Trec$ and stays true for at least $\Tdur$. This is an important difference because our resiliency semantics is defined as a recoverability-durability pair, roughly corresponding to the time it takes to recover from a violation of $\varphi$ and the time for which $\varphi$ subsequently remains true (see Definition~\ref{def:resv}). Since the stability pattern matches multiple recovery episodes, it is not clear which episode to use in the computation of our resiliency semantics.  Our definition solves this ambiguity by considering the first episode. We remark that SRS allows us to express properties like $\mathbf{G}_{I} R_{\Trec,\Tdur}(\varphi)$ (or $\mathbf{F}_{I} R_{\Trec,\Tdur}(\varphi)$), which can be 
understood as enforcing the resiliency requirement $R_{\Trec,\Tdur}(\varphi)$ for each (for some) recovery episode within $I$, and whose semantics can be interpreted as the worst- (best-) case recovery episode within $I$. 
\end{remark}

\begin{remark}[Why a new logic? (2/2)]
The reader might wonder why we would be interested in using temporal operators to reason about a resiliency atom of the form $R_{\Trec,\Tdur}(\varphi)$ as opposed to pushing these operators into $R_{\Trec,\Tdur}(\varphi)$. For example, why would we consider (1)~$\mathbf{G}_{I}\srs_{\Trec,\Tdur}(\varphi)$ instead of (2)~$\srs_{\Trec,\Tdur}(\mathbf{G}_{I} \varphi)$ for an STL formula $\varphi$? The two expressions are fundamentally different: (1)~states that the resiliency specification holds for all $\varphi$-related recovery episodes occurring in interval $I$; (2)~states that the resiliency specification must hold in the first recovery episode relative to $\mathbf{G}_{I} \varphi$ (i.e., the first time $\mathbf{G}_{I} \varphi$ switches from false to true). Arguably, (1) is more useful than (2), even though both are reasonable SRS expressions (compound and atomic, respectively). 
\end{remark}

\subsection{Semantics of Resiliency Specifications}

We provide a quantitative semantics for SRS specifications in the form of a \textit{resilience satisfaction value} (ReSV). Intuitively, an ReSV value quantifies the extent to which recoverability and durability are satisfied. More precisely, it produces a non-dominated set of pairs $(x_r,x_d)\in \mathbb{Z}^2$, where (in the atomic case) $x_r$ quantifies how early before bound $\Trec$ the system recovers, and $x_d$ quantifies how long after bound $\Tdur$ the property is maintained. We further demonstrate the soundness of the ReSV-based semantics w.r.t.\ the STL Boolean interpretation of resiliency specifications. The first step is to establish when one recoverability-durability pair is better than another. 

A \emph{set $S\subseteq \mathbb{R}^n$ of non-dominated tuples} is one where no two tuples $x$ and $y$ can be found in $S$ such that $x$ \textit{Pareto-dominates} $y$, denoted by $x \succ y$. We have that $x \succ y$ if $x_i\geq y_i$, $1\leq i\leq n$, and $x_i>y_i$ for at least one such $i$, under the usual ordering $>$.

We define a novel notion of ``resilience dominance''  captured by the relation $\succ_{re}$ in $\mathbb{Z}^2$.
This is needed because using the standard Pareto-dominance relation  $\succ$ (induced by the canonical $>$ order) would result in an ordering of ReSV pairs that is inconsistent with the Boolean satisfiability viewpoint. Consider the pairs $(-2,3)$ and $(1,1)$.  By Pareto-dominance, $(-2,3)$ and $(1,1)$ are mutually non-dominated, but an ReSV of $(-2,3)$ indicates that the system doesn't satisfy recoverability; namely it recovers two time units too late. On the other hand, an ReSV of $(1,1)$ implies satisfaction of both recoverability and durability bounds, and thus should be preferred to $(-2,3)$. We formalize this intuition next.

\begin{definition}[Resiliency Binary Relations]\label{def:binaryrelation}
We define binary relations $\succ_{re}$, $=_{re}$, and $\prec_{re}$ in $\mathbb{Z}^2$.  Let $x,y\in \mathbb{Z}^2$ with $x=(x_r,x_d)$, $y=(y_r,y_d)$, and \emph{sign} is the signum function. 
We have that $x \succ_{re} y$ if either of the following conditions holds:
\begin{enumerate}
    \item $sign(x_r)+sign(x_d) > sign(y_r)+sign(y_d)$.
    \item $sign(x_r)+sign(x_d) = sign(y_r)+sign(y_d)$, and $x \succ y$.
\end{enumerate}
We say that $x$ and $y$ are \emph{mutually non-dominated}, denoted $x =_{re} y$, if $sign(x_r)+sign(x_d) = sign(y_r)+sign(y_d)$ and neither $x\succ y$ nor $x \prec y$.
Under this ordering, a \emph{non-dominated set} $S$ is such that $x=_{re} y$ for every choice of $x,y \in S$. We denote by $\prec_{re}$ the dual of $\succ_{re}$.
\end{definition}

It is easy to see that $\succ_{re}$, $\prec_{re}$ and $=_{re}$ are mutually exclusive, and in particular, they collectively form a partition of $\mathbb{Z}^2\times \mathbb{Z}^2$. This tells us that for any $x,y \in \mathbb{Z}^2$, either $x$ \emph{dominates} $y$ (i.e., $x\succ_{re}y$), $x$ \emph{is dominated} by $y$ (i.e., $x\prec_{re}y$), or the two are mutually non-dominated (i.e., $x=_{re}y$).

\begin{restatable}{lemma}{order}\label{lemma:largerispartialorder}
Relations $\succ_{re}$ and $\prec_{re}$ are strict partial orders.
\end{restatable}

A proof can be found in Appendix~\ref{appendix:lemma}.  

\begin{definition}[Maximum and Minimum Resilience Sets]\label{def:maxmin_res_set}
Given $P\subseteq\mathbb{Z}^2$, \revision{with $P\neq \emptyset$,} the \emph{maximum resilience set} of $P$, denoted $\max_{re}(P)$, is the largest subset $S\subseteq P$ such that $\forall x\in S$, $\forall y\in P$, $x\succ_{re}y$ or $x=_{re}y$.
The \emph{minimum resilience set} of $P$, denoted $\min_{re}(P)$, 
is the largest subset $S\subseteq P$ such that $\forall x\in S$, $\forall y\in P$, $x\prec_{re}y$ or $x=_{re}y$.
\end{definition}

\begin{restatable}{corollary}{maxminsetone}\label{corollary:maxminsetarenondominated}
Maximum and minimum resilience sets are \revision{non-empty} and non-dominated sets.
\end{restatable}

A proof can be found in Appendix~\ref{appendix:corollary}.

\revision{
\begin{example}
Let $P=\{(-1,2),(1,-2), (2,-1)\}$. Then, we have $\max_{re}(P)=\{(-1,2), (2,-1)\}$ because $(-1,2)=_{re} y$ for all $y\in P$ and  $(2,-1)\succ_{re} (1,-2)$, $(2,-1)=_{re} (1,-2)$, and $(2,-1)=_{re} (2,-1)$. In contrast, $(1,-2)$ is not in $\max_{re}(P)$ because $(1,-2)\prec_{re}(2,-1)$. 
Similarly, we have $\min_{re}(P)=\{(-1,2), (1,-2)\}$. We also note that (as per Corollary~\ref{corollary:maxminsetarenondominated}) the elements of $\max_{re}(P)$ and $\min_{re}(P)$ are mutually non-dominated, i.e., $(-1,2) =_{re} (2,-1)$ and  $(-1,2) =_{re} (1,-2)$, respectively. 
\end{example}
}

Now we are ready to introduce the semantics for our SRS logic. Its definition makes use of maximum and minimum resilience sets in the same way as the traditional STL robustness semantics (see Section~\ref{sec:prelim}) uses max and min operators in compound formulas. Hence, by Corollary~\ref{corollary:maxminsetarenondominated}, our semantics produces non-dominated sets, which implies that all pairs in such a set are equivalent from a Boolean satisfiability standpoint.  This is because $x_r>0$ ($x_d>0$) in Definition~\ref{def:binaryrelation} implies Boolean satisfaction of the  recoverability (durability) portion of an $R_{\Trec,\Tdur}(\varphi)$ expression.
This property will be useful in Theorem~\ref{theorem:soundness}, where we show that our semantics is sound with respect to the Boolean semantics of STL.  

\begin{definition}[Resilience Satisfaction Value]\label{def:resv}
Let $\psi$ be an SRS specification and
$\xi:\mathbb{T}\rightarrow\mathbb{R}^n$ a signal.
We define $\resv(\psi,\xi,t) \subseteq\mathbb{Z}^2$, the \emph{resilience satisfaction value} (ReSV) of $\psi$ with respect to $\xi$ at time $t$, as follows. 
\begin{itemize}
    \item For $\psi$ an SRS atom of the form $\srs_{\Trec,\Tdur}(\varphi)$, $\varphi$ an STL formula,
    \vspace{0.25ex}
    \begin{align}
    \resv(\psi,\xi,t) = \{(-t_{rec}(\varphi,\xi,t)+\Trec, t_{dur}(\varphi,\xi,t)-\Tdur)\}\label{eq:resv}
\end{align}
where 
\begin{align}
    t_{rec}(\varphi,\xi,t) &= \min \left( \{d \in \mathbb{T} \mid (\xi,t+d) \models \varphi \} \cup \{ |\xi|-t \}\right)\label{eq:t_rec}\\[0.25ex]
    t_{dur}(\varphi,\xi,t) &= \min \left( \{d \in \mathbb{T} \mid (\xi,t'+d) \models \neg \varphi \} \cup \{ |\xi|-t' \}\right),\label{eq:t_dur}\\
    &\hspace{2em}t'= t+t_{rec}(\varphi,\xi,t)\nonumber
\end{align}

\item The ReSV of a composite SRS formula is defined inductively as follows.
\vspace{0.25ex}
\begin{align*}
    \resv(\neg\psi,\xi,t) &= \{(-x,-y): (x,y) \in \resv(\psi,\xi,t)\}\\[0.1ex]
    \resv(\psi_1\wedge\psi_2,\xi,t) &= {\min}_{re} (\resv(\psi_1,\xi,t)\;\cup\; \resv(\psi_2,\xi,t))\\[0.1ex]
    \resv(\psi_1\vee\psi_2,\xi,t) &= {\max}_{re} (\resv(\psi_1,\xi,t)\;\cup\; \resv(\psi_2,\xi,t))\\[0.1ex]
    \resv(\mathbf{G}_I\psi,\xi,t) &= {\min}_{re}(\cup_{t'\in t+I} \,\resv(\psi,\xi,t'))\\[0.1ex]
    \resv(\mathbf{F}_I\psi,\xi,t) &= {\max}_{re}(\cup_{t'\in t+I} \,\resv(\psi,\xi,t'))\\[0.1ex]
    \resv(\psi_1\mathbf{U}_I\psi_2,\xi,t) &= {\max}_{re} \cup_{t'\in t+I}{\min}_{re}( \resv(\psi_2,\xi,t')~\cup\\
    &\hspace{7em}{\min}_{re}\cup_{t''\in[t, t+t')}\resv(\psi_1,\xi,t''))
\end{align*}
\end{itemize}
\end{definition}

In the base case, $t'=t+t_{rec}(\varphi,\xi,t)$ is the first time $\varphi$ becomes true starting from and including $t$. If recovery does not occur along $\xi$ (and so the first set in Eq.~\eqref{eq:t_rec} is empty), then $t'=|\xi|$ (the length of the trajectory). 
Similarly, $t'+t_{dur}(\varphi,\xi,t)$ is the first time $\varphi$ is violated after $t'$. Thus $t_{dur}(\varphi,\xi,t)$ quantifies the maximum time duration $\varphi$ remains true after recovery at time $t'$. If $\varphi$ is true for the entire duration of $\xi$, then $t_{rec}(\varphi,\xi,t)=0$ and $t_{dur}(\varphi,\xi,t)=|\xi|-t$. 
If $\varphi$ is false for the duration of $\xi$, then $t_{rec}(\varphi,\xi,t) = |\xi|-t$;
therefore $t'=|\xi|$ and $t_{dur}(\varphi,\xi,t)=0$.

Therefore, the semantics for an SRS atom $\srs_{\Trec,\Tdur}(\varphi)$ is a singleton set $\{(x_r,x_d)\}$, where $x_r$ quantifies how early before time bound $\Trec$ recovery occurs, and $x_d$ indicates for how long after time bound $\Tdur$ the property is maintained. Thus, $x_r$ and $x_d$ quantifies the satisfaction extent (in time) of recoverability and durability, respectively. An important property follows from this observation: similar to traditional STL robustness, a positive (pair-wise) ReSV value indicates satisfaction of recoverability or durability, a negative ReSV value indicates violation, and larger ReSV values indicate better resiliency, i.e., shorter recovery times and longer durability. 

The ReSV semantics for composite SRS formulas is derived by computing sets of maximum/minimum recoverability-durability pairs in a similar fashion to STL robustness: for the $\wedge$ and $\mathbf{G}$ operators, we consider the minimum set over the ReSV pairs resulting from the semantics of the subformulas; for $\vee$ and $\mathbf{F}$, we consider the maximum set. We remark that our semantics induces sets of pairs (rather than unique values); our $\succ_{re}$ and $\prec_{re}$ relations used to compute maximum and minimum resilience sets are therefore partial in nature.
This is to be expected because any reasonable ordering on multi-dimensional data is partial by nature.
For example, given pairs $(2,5)$ and $(3,3)$, there is no way to establish which pair dominates the other: the two are indeed non-dominated and, in particular, the first pair has better durability but worse recoverability than the second pair. In such a situation, our semantics would retain both pairs.

\input{complexity}

\begin{remark}[Bounded-time SRS formulas]\label{rem:bounded_srs}
We say that an SRS formula is \emph{bounded-time} if all of its temporal operators are bounded and the STL formulas serving as SRS atoms are bounded-time.  When a bounded signal $\xi$ is not long enough to evaluate a bounded-time formula, we extend it to the required length simply by repeating its terminal value $\xi(T)$.  For example, let $\varphi = \mathbf{G}_{[0,10]}\;(2\leq y \leq 4)$ and consider $\xi_3$ from Figure~\ref{fig:stl_property}(b), where $|\xi_3|=t_2+\beta$.  We extend $|\xi_3|$ to $t_2+\beta+10$, the required length to determine if $(\xi_3,t+\beta)\models\varphi$ (it does). \revision{Note that the signal extension might result in an overestimate of recoverability or durability.}
\end{remark}

\begin{remark}[Relation to time robustness]\label{rem:time_rob}
A notion of (right) time robustness of an STL atomic proposition $p$ on a trajectory $\xi$ at time $t$ is given by~\cite{donze2010robust,rodionova2021time}:
\vspace{0.1ex}
\begin{align*}
\theta^+(p,\xi,t)=\chi(p, \xi,t)\cdot \max\{d\geq 0\ s.t.\ \forall\,t'\in [t,t+d], \chi(p, \xi,t') = \chi(p, \xi,t)\}
\end{align*}
\noindent where $\chi(p, \xi,t)=+1$ if $(\xi,t)\models p$, and $-1$ otherwise. Intuitively, $|\theta^+(p,\xi,t)|$ measures how long after $t$ property $p$ (or $\neg p$) remains satisfied. 
One might be tempted to use $\theta^+$ in our definition  $t_{rec}(\varphi,\xi,t)$ by setting it to $\max\{0,\theta^+(\neg \varphi,\xi,t)\}$, i.e., the maximum time duration for which $\varphi$ is violated (or $0$ if $\varphi$ holds at $t$). This, however, implies that $t'=t+t_{rec}(\varphi,\xi,t)$ now represents the last time point for which $\varphi$ is false. In our definition, we want instead $t'$ to be the first time point $\varphi$ becomes true. This difference is important, especially for discrete-time signals where the distance between two consecutive time points is non-negligible.

Moreover, time robustness may not quite handle some common corner cases.  Consider a proposition $p$, and two signals $\xi_1$ and $\xi_2$ such that $(\xi_1,t)\models p$ and $(\xi_1,t')\not\models p$, $t'>t$, and $(\xi_2,t)\not\models p$ and $(\xi_2,t')\models p$, $t'>t$. The two signals have opposite behaviors in terms of satisfying $p$. In discrete-time settings (where a discrete-time signal is interpreted in the continuous domain as a piece-wise constant function), we have that $\xi_1$ ($\xi_2$) satisfies $p$ ($\neg p$) throughout the interval $[t,t+1)$ (i.e., for ``almost'' 1 time unit). However, time robustness cannot distinguish between the two signals, namely, $\theta^+(p,\xi_1,t)=\theta^+(p,\xi_2,t)=0$. Thus, if we used $\theta^+$ to define $t_{rec}$ it would be impossible to disambiguate between the first case (where no violation occurs at $t$) and the second case (where a violation occurs at $t$ followed by a recovery episode at the next step). Our definition of $t_{rec}$ in Eq.~\eqref{eq:t_rec} correctly assign a value of~$0$ to $\xi_1$ ($p$ is already satisfied, no recovery at $t$) and a value of~$1$ to $\xi_2$ (from $t$, it takes~$1$ time unit for $p$ to become true).
\end{remark}

\revision{
\begin{remark}
We focus on discrete-times signals, thus a discrete-time SRS framework and discrete-time ReSV semantics. However, we expect that our approach can be extended to continuous time in a straightforward manner,
because STL is well-defined over continuous-time signals. We leave this extension for future work.
\end{remark}
}

\begin{figure}[ht]
	\centering
	\subfloat[]{\includegraphics[width=.53\linewidth]{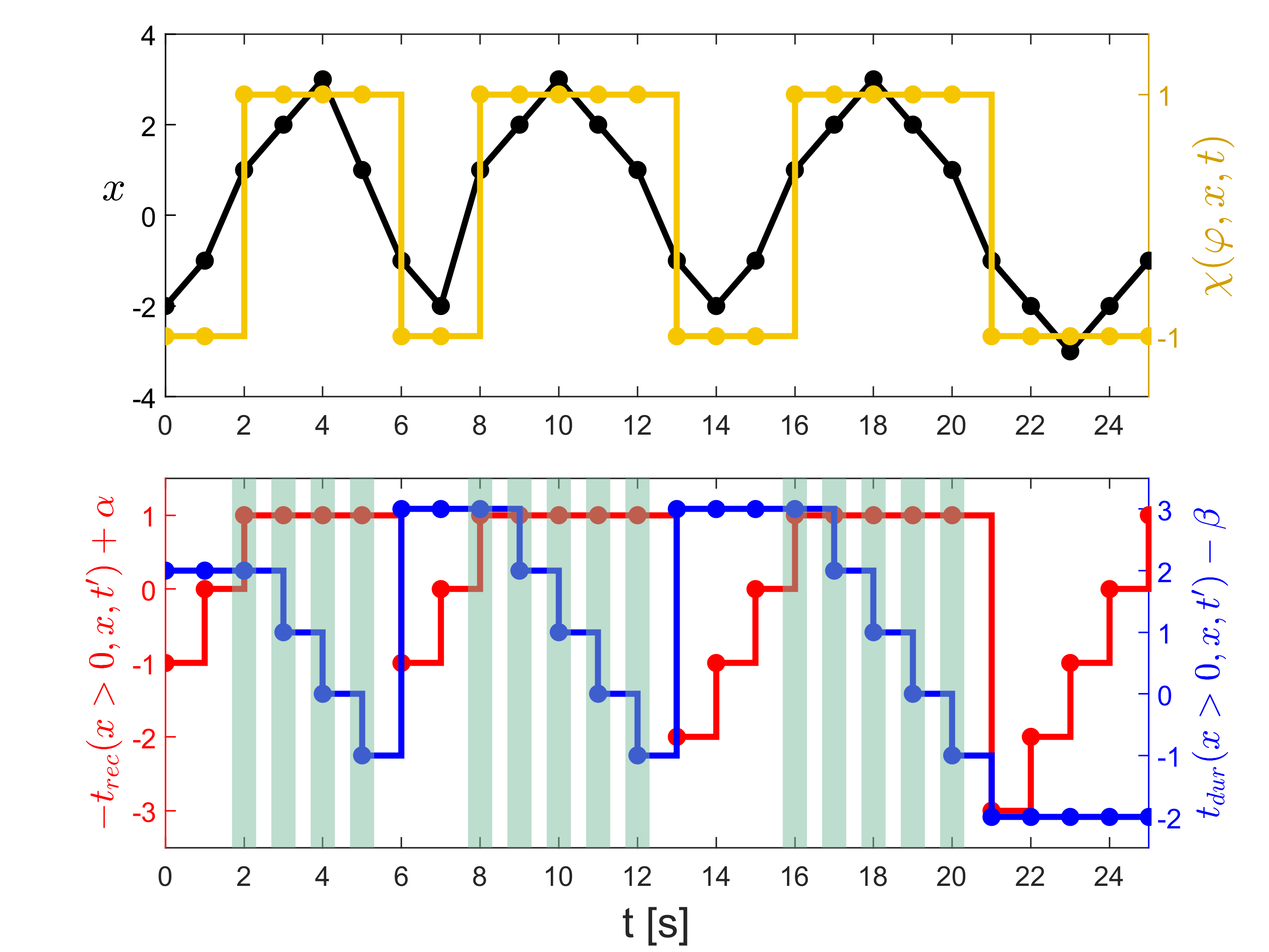}}\hfill
	\subfloat[]{\includegraphics[width=.47\linewidth]{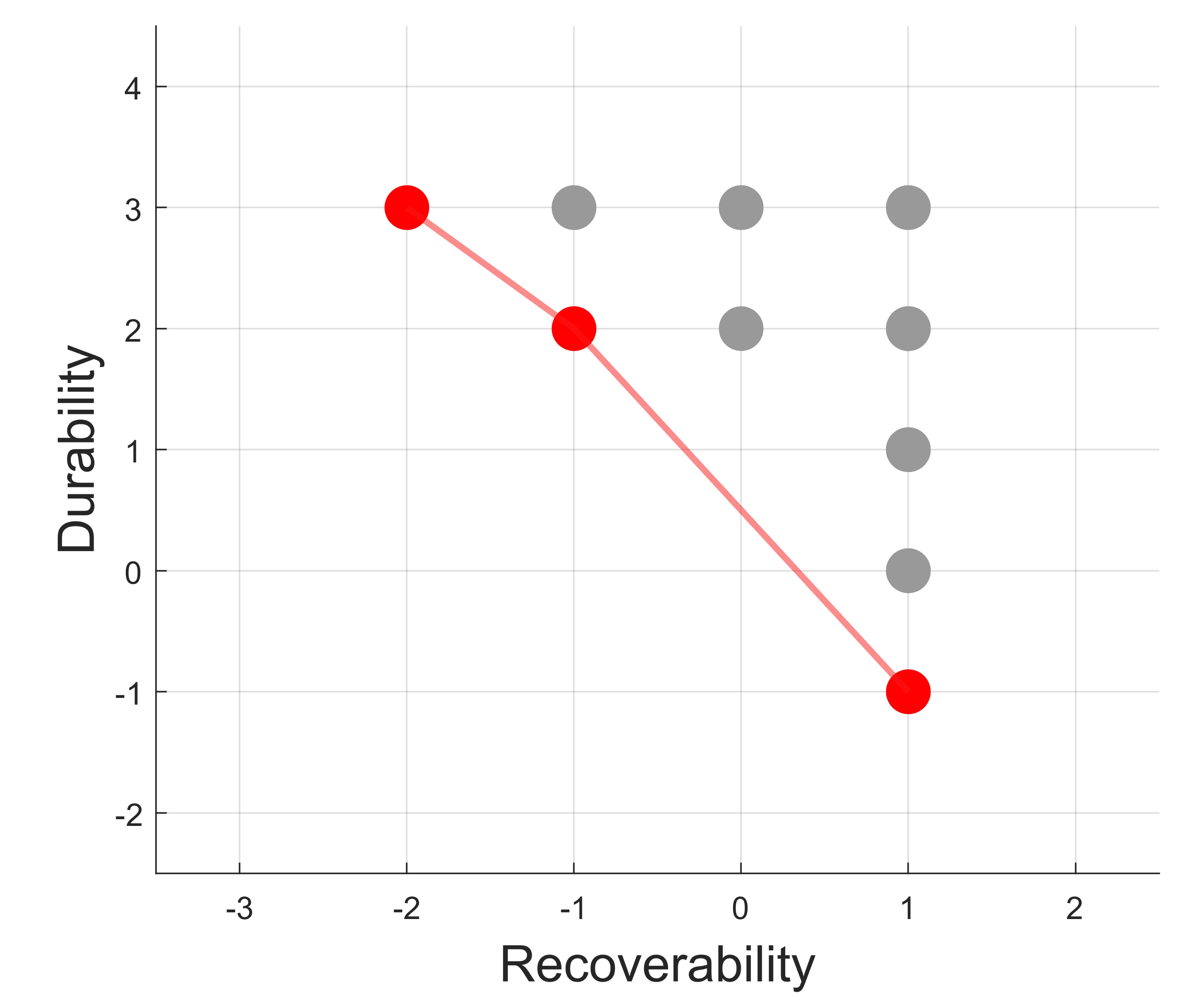}}
	\vspace{-1.5ex}
	\caption{(a)~Signal $x$ and its Boolean semantics w.r.t.\ $x>0$ (top). Recoverability (red) and durability (blue) values for $x>0$ over signal $x$ (bottom). Green bars indicate $x>0$. (b)~Red dots represent the Pareto front for the recoverability-durability pairs of $R_{\Trec,\Tdur}(x>0)$ over time interval $[0,20]$.}
	\vspace{-3ex}
	\label{fig:resv_example}
\end{figure}

\begin{example}
In Figure~\ref{fig:resv_example}, we consider a (one-dimensional) signal $x$ (Figure~\ref{fig:resv_example}(a) top) and the proposition $x>0$ (whose Boolean satisfaction value w.r.t.\ $x$ is plotted on top of $x$). Consider the SRS formula $\psi_1 = \mathbf{G}_{[0,20]}R_{\Trec,\Tdur}(x>0)$ with $\alpha=1$, $\beta=2$. 
Following Definition~\ref{def:resv}, the ReSV of $\psi_1$ can be written as:
\vspace{0.1ex}
\begin{align*}
    \resv(\psi_1,x,0)={\min}_{re} \left(\cup_{t'\in[0,20]}\,\resv(R_{\Trec,\Tdur}(x>0),x,t')\right)
\end{align*}
\noindent
where $ \resv(R_{\Trec,\Tdur}(x>0),x,t') = \{(-t_{rec}(x>0,x,t')+\Trec, t_{dur}(x>0,x,t')-\Tdur)\}$.
Figure~\ref{fig:resv_example}(a) bottom shows how the $-t_{rec}$ and the $t_{dur}$ values evolve over time. These values are also displayed in Figure~\ref{fig:resv_example}(b) on a recoverability-durability plane, to better identify the Pareto-optimal values that constitute the ReSV of $\psi_1$, i.e., the elements of the minimum resilience set in the RHS of the above equation. This is equal to $\{(-1,2),(1,-1),(-2,3)\}$, representing the recoverability and durability values of $x$ w.r.t.\ $R_{\Trec,\Tdur}(x>0)$ at times $t'=0,5,13$.

On the other hand, we obtain a different ReSV if we consider the SRS formula $\psi_2= R_{\Trec,\Tdur}(\mathbf{G}_{[0,20]} x>0)$, where the $\mathbf{G}$ temporal operator is pushed inside the resiliency atom.
Since $x$ never satisfies $\mathbf{G}_{[0,20]} x>0$ (thus, it violates the property at time $0$ and never recovers), we have $t_{rec}(\mathbf{G}_{[0,20]} x>0,x,0)=|x|=25$ and  $t_{dur}(\mathbf{G}_{[0,20]} x>0,x,0)=0$, resulting in $r(\psi_2)=\{(-24,-2)\}$.
\end{example}

\begin{restatable}{proposition}{resvset}
\label{proposition:resvisnondominatedset}
The ReSV $\resv(\psi,\xi,t)$ of an SRS formula $\psi$ w.r.t.\ a signal $\xi$ at time~$t$ is a non-dominated set.
\end{restatable}

A simple proof is given in Appendix~\ref{appendix:proposition},
showing that any ReSV is either a maximum or minimum resilience set and thus, a non-dominated set as per Corollary~\ref{corollary:maxminsetarenondominated}. 
We are now ready to state the main theorem of our work, which establishes the soundness and completeness of SRS logic; i.e., our semantics is consistent from a Boolean satisfiability standpoint. 

\input{theorem_v1}

A proof is given in Appendix~\ref{appendix:theorem}. 
Since every ReSV $\resv(\psi,\xi,t)$ is a non-dominated set, then  $\exists\;x\in\resv(\psi,\xi,t)$ s.t.\ $x\succ_{re}\mathbf{0}$ implies that  $x\succ_{re}\mathbf{0}$ holds for all $x\in\resv(\psi,\xi,t)$ (see the proof of Theorem~\ref{theorem:soundness} for further details).

%% file: complexity.tex
\paragraph{Algorithm for computing the ReSV function $r$.} The algorithm to compute $r$ is a faithful implementation of Definition~\ref{def:resv}. It takes an SRS formula $\psi$, signal $\xi$, and time $t$, produces a syntax tree representing $\psi$, and computes the $r$-values of the subformulas of $\psi$ in a bottom-up fashion starting with the SRS atomic $R$-expressions at the leaves.
Each leaf-node computation amounts to evaluating satisfaction of the corresponding STL formula. The complexity of this operation is $O(|\xi|^{2l})$, for a trajectory $\xi$ and an STL formula with at most $l$ nested \emph{until} operators~\cite{donze2010robust,haghighi2015spatel}. Let $m$ be the number of $R$-expressions in $\psi$. 
Given that we need to evaluate the $R$-expressions at each time point along $\xi$, the time needed to compute the $r$-values of all SRS atoms in $\psi$ is
$O(m\,|\xi|^{2l+1})$. 

Every node $v$ of the tree has an associated result set $P_v$ of $(\mathit{rec}, \mathit{dur})$ pairs. When $v$ is an interior node of the tree, 
$P_v$ is determined, using Definition~\ref{def:resv}, by computing the maximum or minimum resilience set of the result set(s) of $v$'s children.  
The complexity for computing $P_v$ can thus be shown to be quadratic in the size of its input. In particular, the size of the input at the root of $\psi$'s syntax tree is bounded by $O(|\xi|^{2L})$, where $L$ is the maximum number of nested \emph{until} operators in $\psi$.\footnote{This is because the size of an interior node $v$'s input is bounded (in the case of the \emph{until} operator) by $|\xi|^2$ times the sum of the sizes of the result sets of $v$'s children. The size of the root node's input is thus $O(|\xi|^{2L})$.}
Furthermore, the ReSV of the \emph{until} operator can be computed in a manner similar to how STL robustness is computed for the \emph{until} operator.
Therefore, the complexity of computing the root node's ReSV is $O((|\xi|^{2L})^2+|\xi|^{2L})$.
Consequently, the total complexity of computing the ReSV of $\psi$ is $O(m\,|\xi|^{2l+1}+|\xi|^{4L})$.

%% file: theorem_v1.tex
\begin{restatable}[Soundness and Completeness of SRS Semantics]
{theorem}{soundcomposite}
\label{theorem:soundness}
Let $\xi$ be a signal and $\psi$ an SRS specification.  
The following results at time $t$ hold:\\
1) $\exists\;x\in\resv(\psi,\xi,t)$ s.t.\ $x\succ_{re}\mathbf{0} \Longrightarrow (\xi,t)\models \psi$\\
2) $\exists\;x\in\resv(\psi,\xi,t)$ s.t.\ $x\prec_{re}\mathbf{0} \Longrightarrow (\xi,t)\models\neg\psi$\\
3) $(\xi,t)\models \psi \Longrightarrow  \exists\;x\in\resv(\psi,\xi,t)$ s.t.\ $x\succ_{re} \mathbf{0}$ or $x =_{re} \mathbf{0}$\\
4) $ (\xi,t)\models\neg\psi \Longrightarrow\exists\;x\in\resv(\psi,\xi,t)$ s.t.\ $x\prec_{re} \mathbf{0}$ or $x =_{re} \mathbf{0}$
\end{restatable}

%% file: casestudy_v3.tex
\section{Case Study}\label{sec:case}

In this section, we demonstrate the utility of our SRS logic and ReSV semantics on two case studies.
We use Breach~\cite{donze2010breach} for formulating STL formulas and evaluating their STL Boolean semantics.  Experiments were performed on an Intel Core i7-8750H CPU @ 2.20GHz with 16GB of RAM and Windows~10 operating system. Our resiliency framework has been implemented in MATLAB; our implementation along with our case studies can be found in a publicly-available library.\footnote{See  \url{https://github.com/hongkaichensbu/resiliency-specs} }

\subsection{UAV Package Delivery}
\label{subsec:uav}

\begin{figure}[ht]
    \vspace{-4ex}
    \subfloat[]{\includegraphics[width=.45\linewidth]{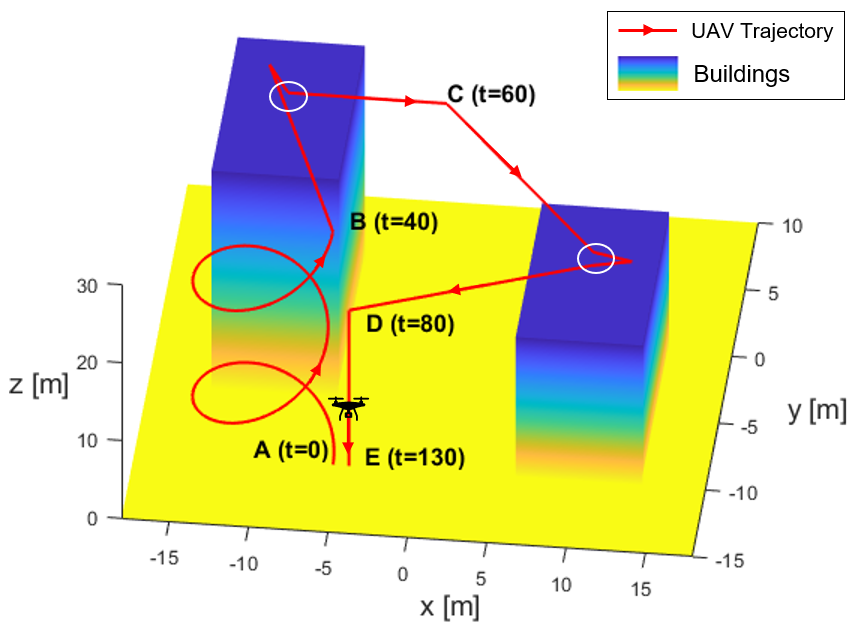}}\hfill
    \subfloat[]{\includegraphics[width=.45\linewidth]{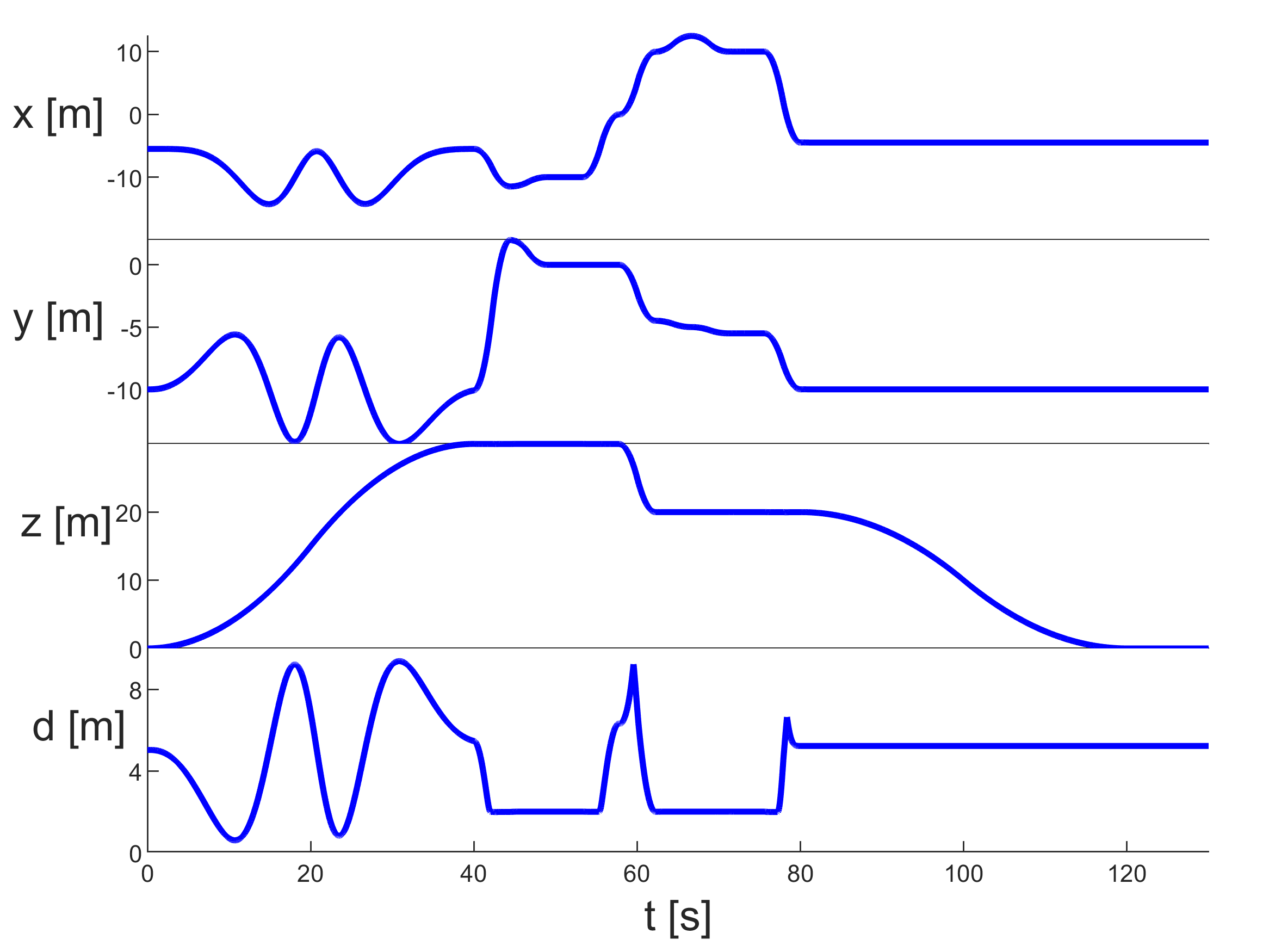}}
    \vspace{-2.5ex}
    \caption{(a)~UAV trajectory and layout of buildings. (b)~UAV coordinates and distance to closest building.}
    \vspace{-5ex}
    \label{fig:uav_case}
\end{figure}
Our first case study involves package delivery via an unmanned aerial vehicle (UAV).  
We consider a quadrotor UAV model with a proportional–integral–deri\-vative~(PID) controller~\cite{michael2010grasp} and use $\mathbf{p} = (x,y,z)$ to denote the UAV's 3-D position vector.  The UAV is tasked with dropping off packages on the rooftops of two adjacent buildings; a simulated trajectory of the UAV along with various points of interest along the trajectory are given in Figure~\ref{fig:uav_case}(a).  As can be seen, the segment of the trajectory from A to B is spiral-like/unstable, and can be attributed to the disturbance caused by the load it is carrying (a disturbance that dissipates once the UAV reaches the building's rooftop).   The coordinates of the UAV and its distance to the closest building are plotted over time in Figure~\ref{fig:uav_case}(b). The (discrete) trajectory is 130 seconds long, i.e., $T=130$ secs, with a time-step of 1/20 secs.  We use the following STL formulas to specify the UAV's mission.\\
[0.5ex]
\emph{Height regulation:} the UAV should remain below 
$H_{max}=120$ meters~\cite{droneheightregulation}. 
\begin{align*}
    \varphi_1 =  \mathbf{G}_{[0,T]}\,\varphi'_1,\hspace{1.5em}  \varphi'_1=(z \leq H_{max})
\end{align*}
\emph{Package delivery:} the UAV needs to hover at a delivery location for a specified period of time (delivery locations are closed $\mathcal{L}^2$-balls with radius $\epsilon=1$ centered at $C_1,C_2$).  The locations are set to $C_1=(-10,0,30)$, $C_2=(10,-5,20)$, and $||\cdot||$ denotes the $\mathcal{L}^2$-norm.
\begin{align*}
\varphi_2 &=  \mathbf{F}_{[0,43]} \mathbf{G}_{[0,1]} \,\varphi'_2,\hspace{1.5em}\varphi'_2 = (||\mathbf{p}-C_1|| \leq \epsilon)\\
\varphi_3 &= \mathbf{F}_{[0,65]} \mathbf{G}_{[0,3]}\,\varphi'_3,\hspace{1.5em} \varphi'_3=(||\mathbf{p}-C_2|| \leq \epsilon)
\end{align*}
\emph{Collision avoidance:} the UAV should maintain a minimum distance $d_{min}=1.5$ meters to the closest building. It is violated repeatedly in the spiral ascent from A to B.
\begin{align*}
\varphi_4 = \mathbf{G}_{[0,T]}\,\varphi'_4,\hspace{1.5em} \varphi'_4= (d\geq d_{min})
\end{align*}

We compute the ReSV values of the above STL formulas when both $\varphi_i$ and $\varphi'_i$ are used as SRS atoms. In particular, in Table~\ref{tab:uav1}, we consider expressions of the form $R_{\Trec,\Tdur}(\varphi_i)$, i.e., where the temporal operators appear inside the SRS atoms, while in Table~\ref{tab:uav2}, nested $\varphi'_i$ expressions are replaced by $R_{\Trec,\Tdur}(\varphi'_i)$ expressions in $\varphi_i$, i.e., temporal operators appear outside the SRS atoms. Syntactically, this difference is
subtle, but as we show, it is of significant importance semantically. 
We assume the UAV state is unchanged after the trajectory ends, when a longer trajectory is needed to determine satisfiability (see Remark~\ref{rem:bounded_srs}).
We choose $\Trec,\Tdur=4$ for the SRS atoms so that we can illustrate repeated property violation and recovery.

\begin{table}[ht]
\vspace{-2ex}
\centering
\begin{tabular}{|>{\centering\arraybackslash}p{3cm}|>{\centering\arraybackslash}p{3cm}|>{\centering\arraybackslash}p{3cm}|}
\hline
{SRS formula} & {$r(\psi_i,\xi,0)$}
& Exec. time (sec) \\ \hline
{$\psi_1 = R_{4,4}(\varphi_1)$} & {\{(4, 126)\} } & {14.69}     \\ \hline
{$\psi_2 = R_{4, 4}(\varphi_2)$}  & {\{(0.1, 45.5)\}}  & {12.74}  \\ \hline
{$\psi_3 = R_{4,4}(\varphi_3)$}   & {\{(-0.2, 65.05)\}}  & {12.53}   \\ \hline
{$\psi_4 = R_{4,4}(\varphi_4)$}    &     {\{(-73.3, 48.7)\}} & {14.18}   \\ \hline
\end{tabular}
\vspace{1ex}
\caption{SRS expressions of the form $R_{\Trec,\Tdur}(\varphi_i)$ for UAV properties $\varphi_i$. All $r$-values are computed using trajectory $\xi$ of Figure~\ref{fig:uav_case}(a) at time~0.}
\vspace{-6ex}
\label{tab:uav1}
\end{table}

\begin{table}[ht]
\centering
\vspace{-3ex}
\begin{tabular}{@{}|c|c|c|c|c|@{}}
\hline
SRS formula & $r(\psi'_i,\xi,0)$ & \,Corresponding SRS atoms\,  & \,Exec.\ time (sec)\,  \\ \hline
$\psi'_1 = \mathbf{G}_{[0,T]} R_{4,4}(\varphi'_1)$ & \{(4, -4)\} & $r(R_{4, 4}(\varphi'_1),\xi,130)$      &  8.50   \\ \hline
$\psi'_2 = \mathbf{F}_{[0,43]} \mathbf{G}_{[0,1]} R_{4,4}(\varphi'_2)$      & \begin{tabular}[c]{@{}c@{}}\{(3.95, -3.5),\\ (4, -3.55)\}\end{tabular}    & \begin{tabular}[c]{@{}c@{}}$r(R_{4, 4}(\varphi'_2),\xi,42.85)$,\\ $r(R_{4,4}(\varphi'_2),\xi,42.95)$\end{tabular}  &  18.84 \\ \hline
$\psi'_3 = \mathbf{F}_{[0,65]} \mathbf{G}_{[0,3]} R_{4,4}(\varphi'_3)$      & \begin{tabular}[c]{@{}c@{}}\{(3.95, -1.35),\\(4, -1.4),\\ (-0.2, 3.05)\}\end{tabular} & \begin{tabular}[c]{@{}c@{}}$r(R_{4, 4}(\varphi'_3),\xi,61.5)$,\\ $r(R_{4,4}(\varphi'_3),\xi,61.6)$,\\ $r(R_{4,4}(\varphi'_3),\xi,65)$\end{tabular}    & 26.70    \\ \hline
$\psi'_4 = \mathbf{G}_{[0,T]} R_{4, 4}(\varphi'_4)$      &  \begin{tabular}[c]{@{}c@{}}\{(-0.25, 5.8),\\(4,-4)\}\end{tabular}   & \begin{tabular}[c]{@{}c@{}}$r(R_{4,4}(\varphi'_4),\xi,8.4)$,\\$r(R_{4,4}(\varphi'_4),\xi,130)$
\end{tabular}    & 8.93
\\ \hline
\end{tabular}
\vspace{1ex}
\caption{Nested $\varphi'_i$ expressions replaced by $R_{\Trec,\Tdur}(\varphi'_i)$ expressions in UAV properties $\varphi_i$. }
\vspace{-4ex}
\label{tab:uav2}
\end{table}

In Table~\ref{tab:uav1}, $r(\psi_1, \xi, 0)=\{(4, 126)\}$ reflects the UAV's resilience w.r.t.\ $\varphi_1$ (height regulation) as it holds in $[0,T]$, thus reaching its maximum $\mathit{rec}$ and $\mathit{dur}$.
Entry $r(\psi_3, \xi, 0)=\{(-0.2, 65.05)\}$ considers the resilience of $\varphi_3$: it is false at time $0$ but becomes true at time $4.2$, making recovery 0.2 secs slower than
$\Trec =4$; 
$\varphi_3$ then remains true until time $73.25$, resulting a durational period 65.05 secs longer than $\Tdur =4$.

Table~\ref{tab:uav2} includes an extra column (the third one) showing the SRS atoms corresponding to each $(\mathit{rec}, \mathit{dur})$ pair in the formula's ReSV (second column). In the first row, the ReSV of $\psi_1'$ is $r(\psi'_1,\xi,0)=\{(4, -4)\}$. Because of the outermost $\mathbf{G}_{[0,T]}$ operator in $\psi'_1$, $(4, -4)$ represents the worst-case recovery episode (relative to the STL property $\varphi'_1$) within the interval $[0,T]$, meaning that $(4, -4)$ is dominated by every other $(\mathit{rec}, \mathit{dur})$ pair in $[0,T]$. The third column tells us that such episode happens at time $t=130$. 
The $(rec,dur)$ values in entry $r(\psi'_3, \xi, 0) = \{(3.95, -1.35), (4, -1.4), (-0.2, 3.05)\}$ represent the best-case episodes~(due to the outermost $\mathbf{F}$ operator) within $t\in [0,65]$ of the ReSV of $\mathbf{G}_{[0,3]} R_{4,4}(\varphi_3')$, which in turn, gives us the worst-case episodes~(due to the inner $\mathbf{G}$ operator) of the ReSV of $R_{4,4}(\varphi_3')$ within $[t,t+3]$.

Even though there are some (relatively small) negative values in Table~\ref{tab:uav2}, overall, the results of Table~\ref{tab:uav2} are consistent with those of  Table~\ref{tab:uav1}, thereby reflecting the overall resiliency of the UAV package-delivery mission. Let us remark the difference between the results of Table~\ref{tab:uav1} and~\ref{tab:uav2}, i.e., between the ReSVs of $\psi_i$ and $\psi'_i$. The former are SRS atoms relative to a composite STL formula $\varphi_i$, and so their ReSVs are singleton sets representing the first recovery episode of $\varphi_i$. The latter are composite SRS formulas relative to an atomic STL formula $\varphi'_i$. As such, their ReSVs are obtained inductively following the structure of the SRS formulas and thus, they may include multiple non-dominated pairs.

We observe that execution times are largely affected by the size of the intervals in the temporal operators appearing in the SRS and STL formulas: computing $\psi_2$ and $\psi_3$ in Table~\ref{tab:uav1} ($\psi'_2$ and $\psi'_3$ in Table~\ref{tab:uav2}) is more efficient than computing $\psi_1$ and $\psi_4$~($\psi'_1$ and $\psi'_4$), even though the former expressions involve nested temporal operators. Indeed, the interval size directly affects both the number of subformula evaluations and the size of the ReSV set. Moreover, our ReSV algorithm uses an implementation of the \emph{always} and \emph{eventually} operators that 
makes them particularly efficient when applied to (atomic) subformulas that exhibit few recovery episodes~(e.g., see the $\mathbf{G}_{[0,T]}$ operator in $\psi'_1$ in Table~\ref{tab:uav2}).

\subsection{Multi-Agent Flocking}

We consider the problem of multi-agent flock formation and maintenance in the presence of external disturbances.  We use the rule-based Reynolds flocking model 
(see Appendix~\ref{appendix:flockingmodel})
involving boids $\mathcal{B}=\{1,\ldots,n\}$ in $m$-dimensional space.  Boid $i$'s position is ${x}_i\in\mathbb{R}^m$, $\mathbf{x} = [x_1,\ldots,x_n]\in\mathbb{R}^{m\cdot n}$ is a global configuration vector, and $\xi =[\mathbf{x}(1),\ldots,\mathbf{x}(k)]\in\mathbb{R}^{m\cdot n\cdot k}$ is a $k$-step trajectory.

We consider a 500-second simulation (trajectory) of a 30-boid flock with a time-step of 0.1 sec in a 2-D plane; so, $T = 500$ secs. 
We subject $20$ of the boids to an intermittent uniformly random displacement~\cite{grosu2020v} from $[0,M]\times[0,2\pi]$, where $M=20$ meters and $2\pi$ are the maximum magnitude and direction of the displacement, respectively. The subset of $20$ boids is chosen uniformly at random during the intervals $[100, 150]$, $[250, 300]$, and $[400, 450]$ in seconds. The simulation starts with the boids at random positions with random velocities, both sampled within some bounded intervals. 

The relevant STL specifications for the flocking mission are the following.\\
[0.5ex]
\emph{Flock formation:} 
a cost function $J(\mathbf{x})$ consisting of a cohesion and a separation term determines whether the boids form a flock~\cite{mehmood2018declarative}:
\begin{align*}
J(\mathbf{x}) = \frac{1}{|\mathcal{B}|\cdot(|\mathcal{B}|-1)} \cdot\sum_{i\in \mathcal{B}}\sum_{j\in \mathcal{B},i<j}||x_{ij}||^2 +\omega\cdot\sum_{(i,j)=\mathcal{E}(\mathbf{x})}\frac{1}{||x_{ij}||^2}
\end{align*}
\noindent where $x_{ij}=x_i-x_j$, $\omega=1/100$, and $\mathcal{E}(\mathbf{x})$ is the set of neighboring boid pairs within an interaction radius $r_c=25$ meters: $\mathcal{E}(\mathbf{x}) = \{(i,j)\in \mathcal{B}^2	\mid||x_{ij}||<r_c, i \not= j\}$.
$J(\mathbf{x})\leq\delta$, $\delta=500$, implies that flock formation has been obtained.
\begin{align*}
    \varphi_1 &= \mathbf{G}_{[0,500]}\,\mathbf{F}_{[0,60]}\,\varphi'_1,\hspace{1.5em} \varphi'_1= (J(\mathbf{x}) \leq \delta)
\end{align*}
Over the whole trajectory, the flock formation should always be obtained in a timely fashion. This is to be expected in the present of recurrent disturbances to the flock.\\
[0.5ex]
\emph{Connected components:} the number of connected components $|CC(\mathbf{x})|$ of the proximity net $\mathcal{G}(\mathbf{x}) = (\mathcal{B}, \mathcal{E}(\mathbf{x}))$ represents potential fragmentation of the flock .  Ideally, $|CC(\mathbf{x})|$ should remain at~1 after flock formation.
\begin{align*}
\varphi_2 &= \mathbf{G}_{[0,500]}\,\mathbf{F}_{[0,60]}\,\varphi'_2,\hspace{1.5em}  \varphi'_2=(|CC(\mathbf{x})|=1)
\end{align*}

\begin{figure}[t]
	\centering
	\includegraphics[width=\linewidth]{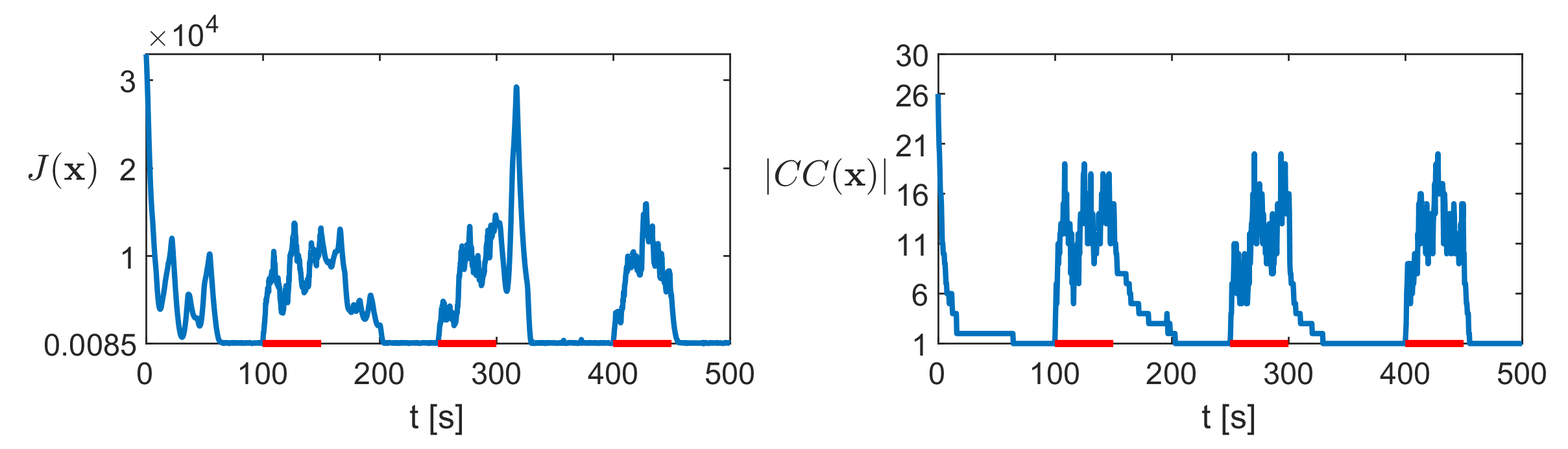}
	\vspace{-6ex}
	\caption{$J(\mathbf{x})$ and $|CC(\mathbf{x})|$ for 30 boids. Red portions of x-axes indicate intervals of random displacement.}
	\label{fig:flock_traj}
\end{figure}

\begin{figure}[t]
	\centering
	\includegraphics[width=.8\linewidth]{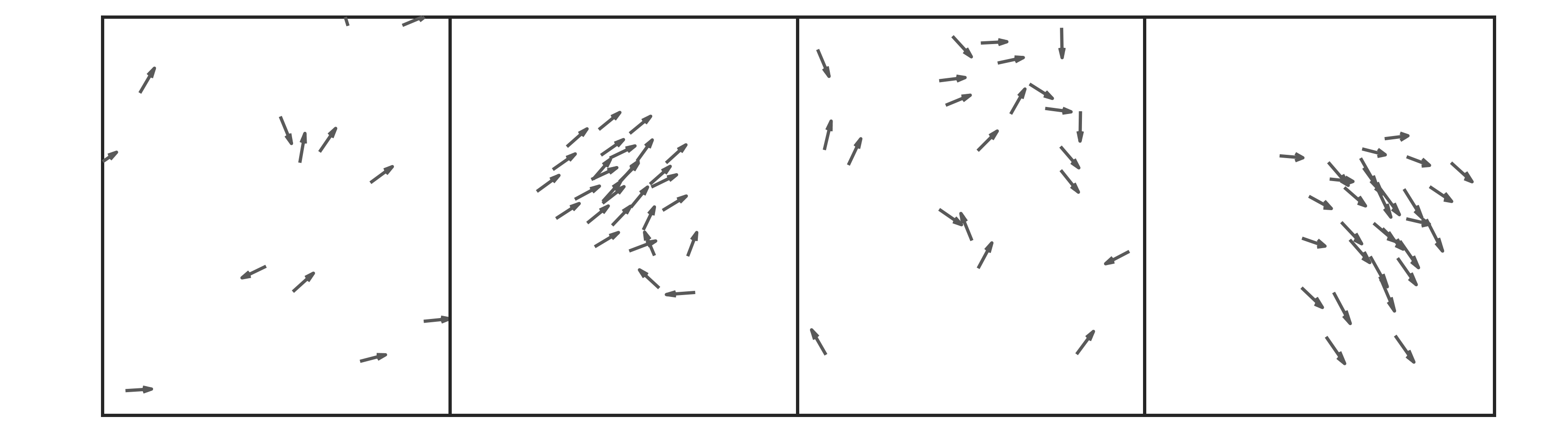}
	\vspace{-3.5ex}
	\caption{Flock simulation snapshots at times $0$, $75$, $300$, and $350$ (left to right).}
	\label{fig:flock_snapshot}
\end{figure}

The functions $J(\mathbf{x})$ and $|CC(\mathbf{x})|$ are plotted in Figure~\ref{fig:flock_traj}. 
Figure~\ref{fig:flock_snapshot} shows flock formation at times $75$ and $350$ where $(J(\mathbf{x}),\;t)\models\varphi'_1$ and $(|CC(\mathbf{x})|,\;t)\models\varphi'_2$, and times $0$ and $300$ where $(J(\mathbf{x}),\;t)\not\models\varphi'_1$ and $(|CC(\mathbf{x})|,\;t)\not\models\varphi'_2$.
Similar to Section~\ref{subsec:uav}, we consider two cases of SRS specifications.
As in Section~\ref{subsec:uav}, we compute the ReSV values of $\varphi_1$, $\varphi_2$ in Table~\ref{tab:flock1} and Table~\ref{tab:flock2}. We select time bounds that are consistent with the timescales of flock formation, i.e., $\Trec, \Tdur = 30$. 

As shown in Table~\ref{tab:flock1}, $r(\psi_2,\xi,0)=\{(-239.3, 200.7)\}$, meaning that $\varphi_2$ is false at time $0$ but becomes true at time $269.3$ (i.e., $239.3$ seconds later than $\Trec=30$) and remains true until the end of the trajectory (i.e., interval $[269.3,T]$ is $200.7$ seconds longer than $\Tdur=30$). 
In Table~\ref{tab:flock2}, $r(\psi'_2,\xi,0)=\{(-42.7, 17.2), (30, -30)\}$; i.e., the $(\mathit{rec}, \mathit{dur})$ pairs at times $130.2$ and $500$, representing the worst recovery episodes in $t\in[0,T]$, which are also the best episodes in the interval $[t,t+60]$. 
Overall, our results show
the resilience of Reynolds flocking model to repeated random disturbances.

\begin{table}[t]
\centering
\begin{tabular}{|>{\centering\arraybackslash}p{3cm}|>{\centering\arraybackslash}p{3cm}|>{\centering\arraybackslash}p{3cm}|}
\hline
SRS formula & $r(\psi_i,\xi,0)$ & Exec. time (sec)  \\ \hline
$\psi_1 = R_{30,30}(\varphi_1)$   & \{(-239.4, 200.6)\}  & 40.47     \\ \hline
$\psi_2 = R_{30,30}(\varphi_2)$ & \{(-239.3, 200.7)\}  &  39.75 
\\ \hline
\end{tabular}
\vspace{1ex}
\caption{SRS expressions of the form $R_{\Trec,\Tdur}(\varphi_i)$ for flocking properties $\varphi_i$. All $r$-values are computed using trajectory $\xi$ of Figure~\ref{fig:flock_traj} at time~0.}
\vspace{-6ex}
\label{tab:flock1}
\end{table}

\begin{table}[t]
\centering
\begin{tabular}{@{}|c|c|c|c|c|@{}}
\hline
SRS formula & $r(\psi'_i,\xi,0)$ &  \,Corresponding SRS atoms\, & \,Exec.\ time (sec)\,  \\ \hline
\begin{tabular}[c]{@{}l@{}}$\psi'_1 =\mathbf{G}_{[0,500]}(\mathbf{F}_{[0,60]}$ \\\hspace{5.3em}$ R_{30,30}(\varphi'_1))$ \end{tabular}  & \begin{tabular}[c]{@{}c@{}}\{(-42.1, 17.8),\\(-19, 13.4),\\(-21.8, 15.3),\\(30,-30)\}\end{tabular} & \begin{tabular}[c]{@{}c@{}}$r(R_{30, 30}(\varphi'_1),\xi,130.4)$,\\$r(R_{30, 30}(\varphi'_1),\xi,280.4)$,\\$r(R_{30, 30}(\varphi'_1),\xi,402.9)$,\\ $r(R_{30, 30}(\varphi'_1),\xi,500)$
\end{tabular}   & 291.72\\ \hline
\begin{tabular}[c]{@{}l@{}}$\psi'_2 =\mathbf{G}_{[0,500]}(\mathbf{F}_{[0,60]}$ \\\hspace{5.3em}$ R_{30,30}(\varphi'_2))$ \end{tabular} & \begin{tabular}[c]{@{}c@{}}\{(-42.7,17.2),\\(30, -30)\}\end{tabular} & \begin{tabular}[c]{@{}c@{}}$r(R_{30, 30}(\varphi'_2),\xi,130.2)$,\\$r(R_{30, 30}(\varphi'_2),\xi,500)$
\end{tabular}      &  279.54   
\\ \hline
\end{tabular}
\vspace{1ex}
\caption{Nested $\varphi'_i$ expressions replaced by $R_{\Trec,\Tdur}(\varphi'_i)$ expressions in flocking properties~$\varphi_i$.}
\vspace{-3ex}
\label{tab:flock2}
\end{table}

%% file: relatedwork_NP.tex
\section{Related Work} \label{sec:relatedwork}

Resiliency has been studied in different system engineering contexts including computer hardware~\cite{hari2017sassifi}, communication~\cite{tan2020new}, distributed systems~\cite{prokhorenko2020architectural}, cyber-security~\cite{yuan2020resilient}, and model checking~\cite{selyunin2017runtime}.

In the context of cyber-physical systems, the literature on resilience is diverse, both in the approach taken and terminology used~\cite{aksaray2021resilient,bouvier2021quantitative,zhu2011robust,mehdipour2021resilience}.
Among the logic-based approaches, standard STL robustness provides a notion of the extent to which a signal can be perturbed in space before affecting property satisfaction, which is seen by some authors as a form of resilience~\cite{mehdipour2021resilience}. STL time robustness~\cite{donze2010robust,rodionova2021time} is the equivalent notion when perturbations in time (forward or backward) are considered. There are similarities between (right) time robustness and our notion of recoverability, but the two semantics are fundamentally different as our ReSVs include a second dimension, durability. See Remark~\ref{rem:time_rob} for more details on this comparison. 
Aksaray et al.~\cite{aksaray2021resilient} propose a (time-) ``shifting'' version of STL and a resilient controller to maximize the robustness value of the shifted formula as fast as possible. This approach, however, 
only supports $\mathbf{G}\mathbf{F}$-formulas and does not provide a dedicated ``shifting'' semantics. 

Control-theoretic characterizations of resiliency includes Bouvier et al.~\cite{bouvier2021quantitative}, which defines a resilience measure (optimized by the controller) based on the additional time it will take for a system to reach a target if under malfunctions. The control framework of Zhu et al.~\cite{zhu2011robust} defines a resilient system as one that can restore its state after extreme events caused by specific perturbation classes. In contrast, we answer the question ``What is resilience in CPS?'' from a temporal-logic perspective. Namely, we provide a syntax for CPS resilience using STL, and a corresponding quantitative semantics. Using our framework in optimal control is not the focus of the present work,  but it is a natural continuation. 


\revision{Another related line of work involves policy and parameter synthesis under multi-objective temporal-logic specifications~\cite{etessami2007multi,chen2013stochastic,calinescu2018efficient,bakhirkin2019paretolib}. Even though we similarly consider multiple requirements (namely recoverability and durability), there is an important difference: we use the satisfaction degrees of recoverability and durability to define the semantics \revision{of} an STL-based logic, with the semantics of composite formulas derived via Pareto optimization. On the other hand, the above mentioned related work adopt, for each requirement, the usual (Boolean, probabilistic, or quantitative) semantics.  

Finally, parametric STL~\cite{asarin2011parametric,bakhirkin2018efficient} has, akin to our work, a set-based semantics. In this case, the semantics is given by the set of parameter evaluations $p$ for which the resulting concrete formula (instantiated with $p$) is satisfied by the given signal. We remark that, in our approach, the specification is fixed, non-parametric, and the set-based semantics arises from the fact that our recoverability-durability pairs might not be directly comparable, i.e., mutually non-dominated.}

%% file: conclusion.tex
\section{Conclusion} \label{sec:conclusion}

In this paper, we presented a logic-based framework to reason about CPS resiliency. We define resiliency of an STL formula $\varphi$ as the ability of the system to recover from violations of $\varphi$ in a timely and durable manner. These requirements represent the atoms of our formally defined
SRS logic, which allows to combine such resiliency statements using temporal and Boolean operators. We also introduced ReSV, the first multi-dimensional semantics for an STL-based logic. Under this semantics, an SRS formula is interpreted as a set of non-dominated $(\mathit{rec},\mathit{dur})$ pairs, which respectively quantify
how quickly the underlying system recovers from a property violation and for how long it satisfies the property thereafter. Importantly, we proved that our ReSV semantics is sound and complete w.r.t.\ the Boolean semantics of STL. 
We illustrated our new resiliency framework with two case studies: UAV package delivery and flock formation. 
Collectively, our results demonstrate the expressive power and flexibility of our framework in reasoning about resiliency in CPS. 

In summary, the contribution of our work is not just establishing theoretical foundations of CPS resiliency but also providing a method to equip temporal logics with multi-dimensional semantics, an approach that in the future could be extended to support arbitrary multi-requirement specifications beyond resiliency.

%% file: acknowledgment.tex
\subsubsection{Acknowledgments.} We thank the anonymous reviewers for their valuable feedback and suggestions for improving the quality of this paper.  Research supported in part by NSF CNS-1952096, CNS-1553273 (CAREER), OIA-2134840, OIA-2040599, CCF-1918225, and CPS-1446832.

%% file: appendix.tex
\newpage
\appendix
\section{Proof of Corollaries}\label{appendix:corollary}



\maxminsetone*
\vspace{-1ex}
\begin{proof}
Let $S$ be the maximum resilience set of set $P$ \revision{with $P\neq \emptyset$}.

\revision{
We prove that $S$ is non-empty by induction.
\begin{enumerate}
    \item Assume $P$ is a singleton set. It is trivial that $S$ is non-empty.
    \item For $P\neq \emptyset$, let $S\neq \emptyset$ be the maximum resilience set of $P$ and $r \in \mathbb{Z}^2$.  Let $S'$ be the maximum resilience set of  $P\cup\{r\}$.  If $r\in P$, it is trivial that $S' = S$, and hence $S'\neq \emptyset$. If $r\not\in P$, then:
    \begin{itemize}
        \item If exists $s \in S$ s.t.\ $s\succ_{re}r$ or $s=_{re}r$, then we have $s\in S'$ (given that $S$ is the maximum resilience set of $P$, then $s\succ_{re}p'$ or $s=_{re}p'$ for all $p'\in P\cup\{r\}$).
        \item Otherwise, we have $s\prec_{re}r$ for all $s\in S$. In this case, for all $p\in P$ and $s\in S$, it holds that either (1)~$p\prec_{re}s$, which, by transitivity of $\prec_{re}$, implies that $p \prec_{re} r$;  or (2)~$p=_{re}s$, in which case  $p\prec_{re}r$ or $p=_{re}r$ (i.e., $p\not\succ_{re}r$): indeed if $p\succ_{re}r$ held, then by transitivity of $\succ_{re}$, we would have $p\succ_{re}s$, which contradicts the assumption. 
        In either case, $r\in S'$.
        \item Therefore, $S'$ is non-empty.
    \end{itemize}
    \item By induction, $S$ is non-empty for any $P\neq \emptyset$.
\end{enumerate}

We then prove that $S$ is non-dominated. Let $s_1,s_2\in S$.
}
\begin{enumerate}
    \item Because $s_1\in S$ and $s_2\in S \subseteq P$, we have $s_1\succ_{re}s_2$ or $s_1=_{re}s_2$.  Similarly, we have $s_2\succ_{re}s_1$ or $s_2=_{re}s_1$. \revision{Therefore,} $s_1=_{re}s_2$.
    \item By definition, $S$ is a non-dominated set.
\end{enumerate}

Minimum resilience sets can be proven to be \revision{non-empty and} non-dominated sets in a similar manner. \qed
\end{proof}



\section{Proof of Propositions}\label{appendix:proposition}

\resvset*
\vspace{-1ex}
\begin{proof}
We show that for any SRS $\psi$ and signal $\xi$, $r(\psi,\xi,t)$ is either a maximum or minimum resilience set \revision{of a non-empty set}. By Corollary~\ref{corollary:maxminsetarenondominated}, this implies that $r(\psi,\xi,t)$ is a \revision{non-empty and} non-dominated set.

If $\psi$ is an SRS atom, then $r(\psi,\xi,t)$ is a singleton set, and thus it is trivially both the maximum and minimum resilience set of itself. When $\psi$ is of the form $\psi_1\wedge\psi_2$, $\psi_1\vee\psi_2$, or $\psi_1\mathbf{U}_I\psi_2$, its semantics $r(\psi,\xi,t)$ is defined as either a maximum or minimum resilience set (see Definition~\ref{def:resv}). The only remaining case is $\psi = \neg \psi_1$. Let $x,y\in\mathbb{Z}^2$ with $x=(x_r,x_d),y=(y_r,y_d)$.  If $r(\psi_1,\xi,t)$ is a non-dominated set, then $r(\psi,\xi,t)=\{(-x_r,-x_d) \mid x \in r(\psi_1,\xi,t) \}$ is a non-dominated set also. This is because $x=_{re} y$ iff $(-x_r,-x_d)=_{re} (-y_r,-y_d)$ for any $x,y$.
\qed
\end{proof}

\section{Proof of Lemmas}
\label{appendix:lemma}

\order*
\vspace{-1ex}
\begin{proof}
Let $x,y,z \in \mathbb{Z}^2$ be different with $x=(x_r,x_d)$, $y=(y_r,y_d)$, $z= (z_r,z_d)$.

We prove that $\succ_{re}$ has irreflexivity and transivity properties.
\vspace{-2ex}
\begin{itemize}
    \item Irreflexivity: we can see that $x \succ_{re} x$ never holds because $sign(x_r)+sign(x_d) = sign(x_r)+sign(x_d)$ but $x$ does not Pareto-dominate itself.
    \item Transivity: 
    Assume that $x\succ_{re}y$ and $y\succ_{re}z$. Using the transivity of Pareto-dominance $\succ$~\cite{Deb:2001:MOU:559152},
    one of the two following statements is true:
    \begin{itemize}
        \item $sign(x_r)+sign(x_d) > sign(z_r)+sign(z_d)$.
        \item $sign(x_r)+sign(x_d) = sign(z_r)+sign(z_d)$ and $x$ Pareto-dominates $z$.
    \end{itemize}
    Therefore $x\succ_{re}z$.
\end{itemize}
\vspace{-2ex}
Irreflexivity and transitivity together imply asymmetry; so $\succ_{re}$ is a strict partial order, and the dual of a strict partial order, $\prec_{re}$, is also a strict partial order~\cite{davey2002introduction}.\qed
\end{proof}

\section{Proof of Theorems}\label{appendix:theorem}

\soundcomposite*

\input{theoremproof_v1}

\section{Reynolds flocking model}\label{appendix:flockingmodel}

We use Reynolds rule-based model~\cite{reynolds1987flocks,reynolds1999steering} to describe the dynamics of flocking.  There are three steering behaviors used to determine how the agents, which are called boids, maneuver based on the positions and velocities of their nearby flockmates.   \emph{Separation} causes a boid to move away from its nearby flockmates in crowded situations.  \textit{Cohesion} drives a boid to move towards the average location of its nearby flockmates.  \textit{Alignment} causes a boid to steer in the direction of the average heading of its nearby flockmates. 

\begin{figure}[h]
	\centering
	\subfloat[Separation.]{\includegraphics[width=.3\linewidth]{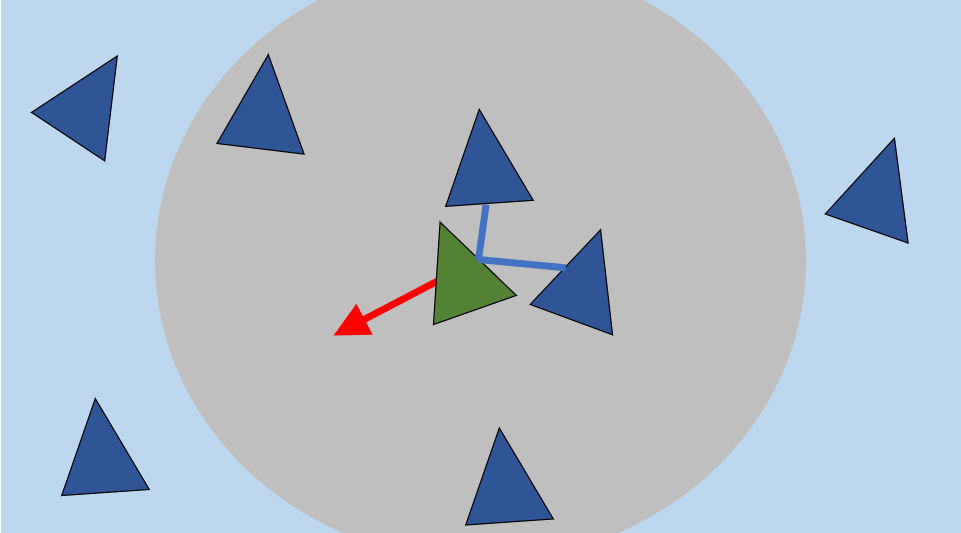}}\hfill
	\subfloat[Cohesion.]{\includegraphics[width=.3\linewidth]{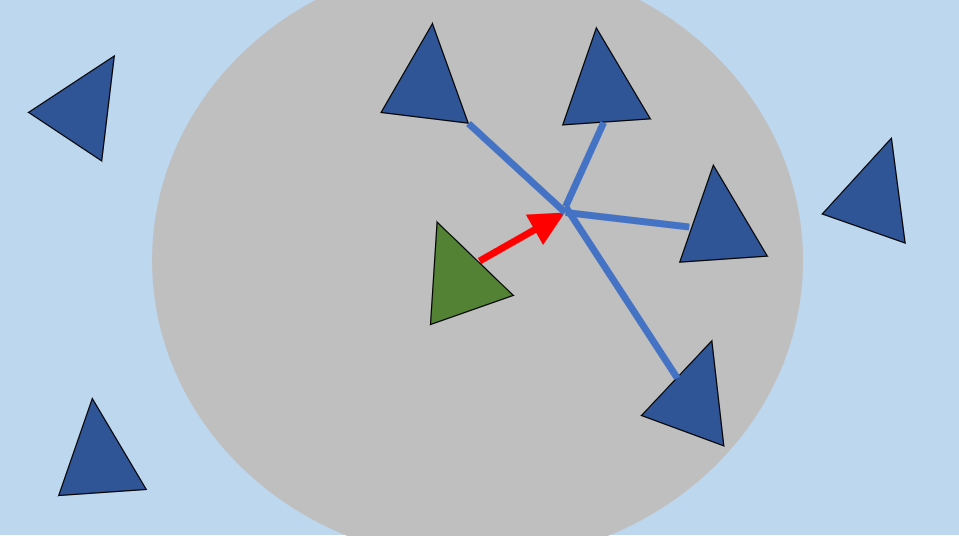}}\hfill
	\subfloat[Alignment.]{\includegraphics[width=.3\linewidth]{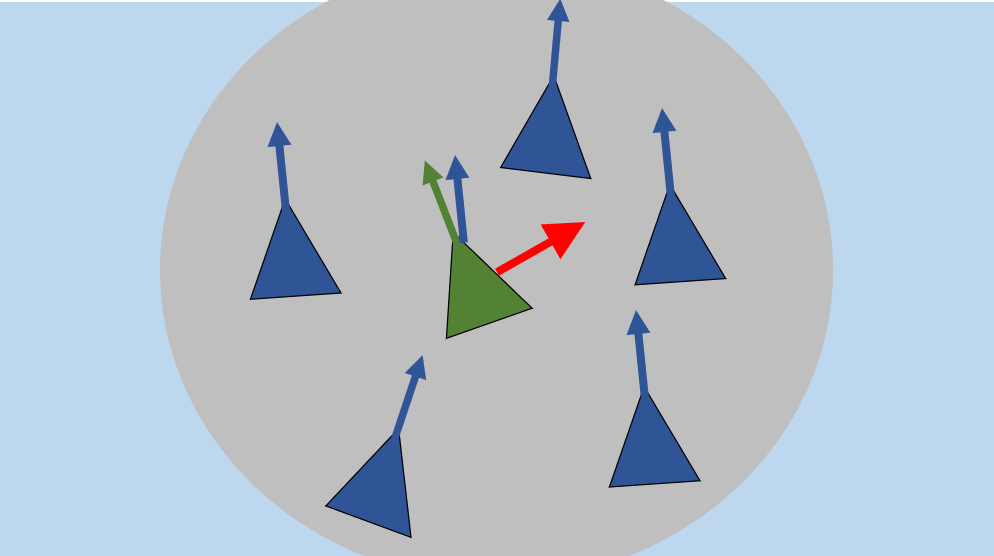}}
	\caption{Reynolds model rules.  Gray circles represent the interaction region of the green boid.}
	\label{fig:flock_rule}
\end{figure}

Figure~\ref{fig:flock_rule} illustrates these rules.  We consider a set of boids $\mathcal{B}= \{1,\ldots,n\}$ in $m$-dimensional space that move in discrete time as follows:
\begin{align*}
    \mathbf{x}(k+1) &= \mathbf{x}(k) + dt \cdot v_i(k) + d_i(k)\\
    v_i(k+1) &= v_i(k) + dt \cdot a_i(k)
\end{align*}
\noindent where $\mathbf{x}(k) = [x_1(k),\ldots,x_n(k)]\in\mathbb{R}^{m\cdot n}$ is the configuration of all boids, $\mathbf{v}(k) = [v_1(k),\ldots,v_n(k)]\in\mathbb{R}^{m\cdot n}$ and $\mathbf{a}(k) = [a_1(k),\ldots,a_n(k)]\in\mathbb{R}^{m\cdot n}$.  Vectors $x_i(k)$, $v_i(k)$, $a_i(k)$, $d_i(k)$ $\in \mathbb{R}^m$ are the position, velocity, acceleration, and random displacement of boid $i$ at step $k$.  Time interval $dt$ is the duration of a time step.

%% file: theoremproof_v1.tex
\begin{proof}
We first establish the following results because they are used in the main proof. Let $x,y\in \mathbb{Z}^2$ with $x=(x_r,x_d)$, $y=(y_r,y_d)$.
\begin{enumerate}[labelindent=\parindent, leftmargin=*, label=\roman*), labelsep=0.2em, widest=viii,
align=left]
    \item $x\succ_{re}\mathbf{0}$ is equivalent to $x_r,x_d\geq 0 \wedge x\not=\mathbf{0}$.
    \item $x=_{re}\mathbf{0}$ is equivalent to $x_r\,x_d< 0 \vee x=\mathbf{0}$.
    \item By~i) and~ii), $x\succ_{re}\mathbf{0} \vee x=_{re}\mathbf{0}$ is equivalent to $x_r,x_d \geq 0 \vee x_r\,x_d< 0$.
    \item If $x\succ_{re}\mathbf{0}$ and  $y\succ_{re}x \vee y=_{re}x$, then $y_r,y_d\geq 0\wedge y\not=\mathbf{0}$ holds, which, by~i), is equivalent  to  $y\succ_{re}\mathbf{0}$.
    \item If $x\succ_{re}\mathbf{0} \vee x=_{re}\mathbf{0}$ and $y=_{re}x$, then it holds that $y_r,y_d\geq 0 \vee y_r\,y_d< 0$, which by iii), is equivalent to $y\succ_{re}\mathbf{0} \vee y=_{re}\mathbf{0}$.
    \item If $x\succ_{re}\mathbf{0} \vee x=_{re}\mathbf{0}$ and $y\succ_{re}x$, then it holds that $y_r,y_d> 0 \vee y_r\,y_d< 0 \vee (y_r=0\wedge y_d>0) \vee (y_r>0 \wedge y_d=0)$ is true, thus $y\succ_{re}\mathbf{0}$ or  $y=_{re}\mathbf{0}$.\footnote{Given $x=_{re}\mathbf{0}$, the relation $y\succ_{re}x$ implies that 1) $y_r\,y_d<0\wedge y\succ x$ holds (when $x_r\,x_d<0$) or 2) $y_r,y_d> 0 \vee (y_r=0 \wedge y_d>0) \vee (y_r>0\wedge y_d=0)$ holds (when $x_r\,x_d<0$ or $x_r,x_d=0$). The disjunction of 1) and 2) is equivalent to $(y_r\,y_d<0\wedge y\succ x) \vee (y_r,y_d\geq 0 \wedge y\succ \mathbf{0})$, which implies $y_r\,y_d<0 \vee y_r,y_d\geq 0$, i.e., $y\succ_{re}\mathbf{0} \vee y=_{re}\mathbf{0}$.}
    \item Let $\varphi$ be an STL formula with $\Trec$,$\Tdur>0$; let $\psi = \srs_{\Trec,\Tdur}(\varphi)$ be an SRS atom. Since $(\xi,t)\models \psi$ is equivalent from a Boolean satisfaction standpoint to $\neg\varphi\mathbf{F}_{[0,\Trec]}\mathbf{G}_{[0,\Tdur)}\varphi$,  we can conclude that $(\xi,t)\models \psi$ holds iff $\exists\; 0 \leq t_1 \leq \Trec$ s.t.\ $\mathbf{G}_{[t,t+t_1)}\neg \varphi \wedge \mathbf{G}_{[t+t_1,t+t_1+\Tdur)}\;\varphi$ is true.
    \item Assume $\exists\,x\in\resv(\psi,\xi,t)$ s.t.\ $x\succ_{re}\mathbf{0}$. By Proposition~\ref{proposition:resvisnondominatedset}, we have  $y=_{re}x$ holds for all $y\in\resv(\psi,\xi,t)$. Then we have $y_r,y_d\geq 0 \wedge y\not=\mathbf{0}$, i.e., $y\succ_{re}\mathbf{0}$.\footnote{Indeed, for $y=_{re}x$ to hold, we need $sign(y_d)+sign(y_r)=sign(x_d)+sign(x_r)\geq 1$, which can only happen if  $y_r,y_d\geq 0 \wedge y\not=\mathbf{0}$.} Similarly, we can prove that $\exists\,x\in\resv(\psi,\xi,t)$ s.t.\ $x=_{re}\mathbf{0}$ implies that $x=_{re}\mathbf{0}$ holds for all $x\in\resv(\psi,\xi,t)$.
\end{enumerate}

We prove statements~1) and~3) of Theorem~\ref{theorem:soundness}. Statements~2) and~4) can be proven analogously. The proofs follow the inductive structure used in the definition of ReSV. We first prove statement~1) of Theorem~\ref{theorem:soundness} as follows.\\[0.75ex]
1) $\exists\;x\in\resv(\psi,\xi,t)$ s.t.\ $x\succ_{re}\mathbf{0} \Longrightarrow (\xi,t)\models \psi$
\begin{itemize}
    \item $\psi = \srs_{\Trec,\Tdur}(\varphi)$.   By Definition~\ref{def:resv}, $\resv(\psi,\xi,t)$ is a singleton set; $x\succ_{re}\mathbf{0}$ implies $x_r,x_d\geq 0\wedge x\not=\mathbf{0}$, which then implies $x_r,x_d\geq 0$.  By transitivity of implication, we can equivalently prove that $x_r,x_d\geq 0$ implies $(\xi,t)\models \psi$.  
    By Eq.~\eqref{eq:resv}, $x_r\geq 0$ implies $0\leq t_{rec}(\varphi,\xi,t)\leq \Trec$. Let $t_1=t_{rec}(\varphi,\xi,t)$. 
    By Eq.~\eqref{eq:t_rec},  we have that $\mathbf{G}_{[t,t+t_1)}\neg\varphi$ holds. By Eq.~\eqref{eq:resv}, 
    $x_d\geq 0$ implies  $t_{dur}(\varphi,\xi,t+t_1)\geq \Tdur$, and thus, by  Eq.~\eqref{eq:t_dur}, $\mathbf{G}_{[t+t_1,t+t_1+\Tdur)}\;\varphi$ holds.  Therefore $x_r,x_d\geq 0$ implies $\exists\,0 \leq t_1 \leq \Trec$ s.t.\ $\mathbf{G}_{[t,t+t_1)}\neg \varphi \wedge \mathbf{G}_{[t+t_1,t+t_1+\Tdur)}\;\varphi$; by result~vii), $(\xi,t)\models \psi$.
    \\[-1.75ex]
    \item $\psi = \neg \psi_1$. By Definition~\ref{def:resv}, $\resv(\psi,\xi,t) = \{(-x_r,-x_d): x \in \resv(\psi_1,\xi,t)\}$. Thus, by the implicant in~i), we have that $\exists\,x \in\resv(\psi_1,\xi,t)$ s.t.\ $x\prec_{re}\mathbf{0}$ holds.
    From the induction hypothesis, we have $(\xi,t)\models\neg\psi_1$; thus $(\xi,t)\models \psi$. 
    \\[-1.75ex]
    \item $\psi = \psi_1\wedge\psi_2$.  By Definition~\ref{def:resv}, $\resv(\psi_1\wedge\psi_2,\xi,t) = {\min}_{re} (\resv(\psi_1,\xi,t)\;\cup\; \resv(\psi_2,\xi,t))$.  By Definition~\ref{def:maxmin_res_set} and result~iv), we have $x'\succ_{re}x$ or  $x'=_{re}x$, thus $x'\succ_{re}\mathbf{0}$ for all $x'\in\resv(\psi_1,\xi,t)\cup\resv(\psi_2,\xi,t)$.  From the induction hypothesis, we have $(\xi,t)\models \psi_1$ and $(\xi,t)\models \psi_2$; 
    thus $(\xi,t)\models \psi_1\wedge\psi_2$.
    \\[-1.75ex]
    \item $\psi = \psi_1\vee\psi_2$.  By Definition~\ref{def:resv}, $\resv(\psi_1\vee\psi_2,\xi,t) = {\max}_{re} (\resv(\psi_1,\xi,t)\;\cup\; \resv(\psi_2,\xi,t))$.  By Definition~\ref{def:maxmin_res_set}, we have $x \in \resv(\psi_1,\xi,t)\;\cup\; \resv(\psi_2,\xi,t)$; therefore $x \in \resv(\psi_1,\xi,t)$ or $x \in \resv(\psi_2,\xi,t)$.  From the induction hypothesis, we have $(\xi,t)\models \psi_1$ or $(\xi,t)\models \psi_2$; thus  $(\xi,t)\models \psi_1 \vee \psi_2$.
    \\[-1.75ex]
    \item $\psi =\psi_1\mathbf{U}_I\psi_2$. By Definition~\ref{def:resv},  $\resv(\psi_1\mathbf{U}_I\psi_2,\xi,t) = {\max}_{re} \cup_{t'\in t+I}{\min}_{re}( \resv(\psi_2,\xi,t')\allowbreak\cup\,{\min}_{re}\cup_{t''\in[t, t+t')}\resv(\psi_1,\xi,t''))$.  By  Definition~\ref{def:maxmin_res_set} and result~iv), we have that $\exists\,t_1\in t+I$ s.t.\ $x_1\succ_{re}\mathbf{0}$ for all $x_1\in \resv(\psi_2,\xi,t_1)\cup\,{\min}_{re}\cup_{t''\in[t, t+t_1)}\resv(\psi_1,\xi,t'')$.  
    From the induction hypothesis, we have $(\xi,t_1)\models \psi_2$.  Similarly, we have $x_2\succ_{re}\mathbf{0}$ for all $x_2\in \resv(\psi_1,\xi,t'')$ and all $t''\in[t,t+t_1)$.
    From the induction hypothesis, $\forall\,t''\in[t,t+t_1)$ $(\xi,t'')\models \psi_1$. Together with $(\xi,t_1)\models \psi_2$ we can conclude that $(\xi,t)\models \psi_1\mathbf{U}_I\psi_2$.
\end{itemize}

We then prove statement~3) of Theorem~\ref{theorem:soundness} as follows.\\[0.75ex]
\noindent 3) $(\xi,t)\models \psi \Longrightarrow \exists\;x\in\resv(\psi,\xi,t)$ s.t.\ $x\succ_{re}\mathbf{0}$ or $x=_{re}\mathbf{0}$\\[-3ex]
\begin{itemize}
    \item $\psi = \srs_{\Trec,\Tdur}(\varphi)$.  By result~(vii),  $(\xi,t)\models \psi$ implies that $\exists\,0 \leq t_1 \leq \Trec$ s.t.\ $\mathbf{G}_{[t,t+t_1)}\neg \varphi \wedge \mathbf{G}_{[t+t_1,t+t_1+\Tdur)}\;\varphi$ is true, which results in $x_r,x_d\geq 0$ by Eqs.~\eqref{eq:t_rec} and~\eqref{eq:t_dur}. This is equivalent to $x\succ_{re}\mathbf{0}$ (when at least one of $x_r,x_d\geq 0$ holds strictly) or $x=_{re}\mathbf{0}$ (when $x_r=x_d= 0$). 
    \\[-1.75ex]
    \item $\psi = \neg \psi_1$.  From the induction hypothesis, $(\xi,t)\models \neg\psi_1$ implies that $\exists\,x\in \resv(\neg \psi_1,\xi,t)$ s.t.\ $x\succ_{re}\mathbf{0}$ or $x=_{re}\mathbf{0}$. Then, by Definition~\ref{def:resv}, we have that $\exists\,x\in \resv(\psi_1,\xi,t)$ s.t.\ $x\prec_{re}\mathbf{0}$ or $x=_{re}\mathbf{0}$
    (i.e., $x_r,x_d\leq 0 \vee x_r\,x_d< 0$ is true).  By Definition~\ref{def:resv}, $(-x_r,-x_d)\in\resv(\psi,\xi,t)$; thus $(-x_r,-x_d)\succ_{re}\mathbf{0}$ (when at least one of $x_r,x_d\leq 0$ holds strictly) or $(-x_r,-x_d)=_{re}\mathbf{0}$ (when $x_r=x_d = 0$ or $x_r\,x_d< 0$ is true).
    \\[-1.75ex]
    \item $\psi = \psi_1\wedge\psi_2$.  This is equivalent to $(\xi,t)\models \psi_1 \wedge (\xi,t)\models \psi_2$ holds.  From induction hypothesis, we have  $\exists\,x_1\in\resv(\psi_1,\xi,t)$ s.t.\ $x_1\succ_{re}\mathbf{0}$ or $x_1=_{re}\mathbf{0}$; and $\exists\,x_2\in\resv(\psi_2,\xi,t)$ s.t.\ $x_2\succ_{re}\mathbf{0}$ or $x_2=_{re}\mathbf{0}$. By result~viii), we have $x\succ_{re}\mathbf{0}$ or $x=_{re}\mathbf{0}$ holds for all $x\in\resv(\psi_1,\xi,t)$ and $x\in\resv(\psi_2,\xi,t)$.  By Definition~\ref{def:maxmin_res_set},  we have $x\succ_{re}\mathbf{0}$ or $x=_{re}\mathbf{0}$ holds for all $x\in\min_{re}(\resv(\psi_1,\xi,t)\cup \resv(\psi_2,\xi,t))$, i.e., by Definition~\ref{def:resv}, for all $x \in \resv(\psi,\xi,t)$.
    \\[-1.75ex]
    \item $\psi = \psi_1\vee\psi_2$.  By definition, $(\xi,t)\models \psi_1\vee\psi_2$ holds implies that $(\xi,t)\models \psi_1 \vee (\xi,t)\models \psi_2$ holds.  From induction hypothesis, we have  $\exists\,x_1\in\resv(\psi_1,\xi,t)$ s.t.\ $x_1\succ_{re}\mathbf{0}$ or $x_1=_{re}\mathbf{0}$; or  $\exists\,x_2\in\resv(\psi_2,\xi,t)$ s.t.\ $x_2\succ_{re}\mathbf{0}$ or $x_2=_{re}\mathbf{0}$.  Let w.l.o.g. the former holds; then we have  $x_1\in\resv(\psi_1,\xi,t)\cup\resv(\psi_2,\xi,t)$.  By Definition~\ref{def:maxmin_res_set}, $x\succ_{re}x_1$ or $x=_{re}x_1$ for all $x\in\max_{re}(\resv(\psi_1,\xi,t)\cup\resv(\psi_2,\xi,t))$.  By results~v) and~vi), $x\succ_{re}\mathbf{0}$ or $x=_{re}\mathbf{0}$. By Definition~\ref{def:resv}, $\resv(\psi,\xi,t)=\max_{re}(\resv(\psi_1,\xi,t)\cup \resv(\psi_2,\xi,t))$.    
    \\[-1.75ex]
    \item $\psi =\psi_1\mathbf{U}_I\psi_2$.  By definition, $(\xi,t)\models\psi_1\mathbf{U}_I\psi_2$ holds implies that $\exists\,t'\in t+I$ s.t.\ $(\xi,t')\models\psi_2 \wedge \forall\,t''\in [t,t'), (\xi,t'')\models\psi_1$.   From induction hypothesis, we have $\exists\,x_2\in\resv(\psi_2,\xi,t')$ s.t.\ $x_2\succ_{re}\mathbf{0}$ or $x_2=_{re}\mathbf{0}$; and $\exists\,x_1\in\resv(\psi_1,\xi,t'')$ s.t.\ $x_1\succ_{re}\mathbf{0}$ or $x_1=_{re}\mathbf{0}$ for all $t''\in [t,t')$.  By result~(viii), we have that $x\succ_{re}\mathbf{0}$ or $x=_{re}\mathbf{0}$ holds for all $x\in \resv(\psi_2,\xi,t')$ and  for all $x\in\resv(\psi_1,\xi,t'')$, for all $t''\in [t,t')$.  By Definition~\ref{def:maxmin_res_set}, we evince that $\exists\,t'\in t+I$ s.t.\ $x\succ_{re}\mathbf{0}$ or $x=_{re}\mathbf{0}$ holds for all $x\in {\min}_{re}( \resv(\psi_2,\xi,t')\cup{\min}_{re}\cup_{t''\in[t, t+t')}\resv(\psi_1,\xi,t''))$. Note that the latter is a subset of $\cup_{t'\in t+I}\min_{re}( \resv(\psi_2,\xi,t')\cup{\min}_{re}\cup_{t''\in[t, t+t')}\resv(\psi_1,\xi,t''))$, which implies, together with  Definition~\ref{def:maxmin_res_set} and results~v) and~vi), that $x\succ_{re}\mathbf{0}$ or $x=_{re}\mathbf{0}$ holds for all $x\in {\max}_{re}\cup_{t'\in t+I}\min_{re}( \resv(\psi_2,\xi,t')\cup{\min}_{re}\cup_{t''\in[t, t+t')}\resv(\psi_1,\xi,t''))$, which is $\resv(\psi,\xi,t)$ by Definition~\ref{def:resv}.
\end{itemize}
This completes the proof.\qed
\end{proof}